\documentclass[letterpaper,12pt]{article}
\usepackage{natbib} \bibpunct{(}{)}{;}{a}{}{,}
\usepackage{fullpage}
\usepackage{amsmath,amssymb,amsthm}
\usepackage{enumerate}
\usepackage{appendix}
\usepackage{graphicx}
\usepackage{color}
\usepackage{hyperref}
\usepackage{xcolor}
\usepackage{comment,stackengine}
\hypersetup{
        colorlinks   = true,
        linkcolor    = blue,
        citecolor    = green
}
\usepackage{authblk}
\newtheorem{thm}{Theorem}[section]

\newtheorem{defn}[thm]{Definition}
\newtheorem{rem}[thm]{Remark}
\newtheorem{ass}[thm]{Assumption}
\newtheorem{lem}[thm]{Lemma}
\usepackage{todonotes}

\newtheorem{theorem}{Theorem}

\newcommand{\E}{\mathbb{E}}
\newcommand{\R}{\mathbb{R}}

\newcommand{\bbP}{\mathbb{P}}

\newcommand{\bbZ}{\mathbb{Z}}

\newcommand{\lambdahat}{\widehat{\lambda}}
\newcommand{\Deltahat}{\widehat{\Delta}}
\newcommand{\Bhat}{\widehat{B}}
\newcommand{\Lhat}{\hat{L}}
\newcommand{\Phat}{\hat{P}}

\newcommand{\Ltilde}{\tilde{L}}

\newcommand{\calB}{\mathcal{B}}
\newcommand{\calD}{\mathcal{D}}
\newcommand{\calI}{\mathcal{I}}
\newcommand{\calL}{\mathcal{L}}
\newcommand{\calM}{\mathcal{M}}

\newcommand{\calT}{\mathcal{T}}

\newcommand{\SBM}{\operatorname{SBM}}
\newcommand{\HSBM}{\operatorname{HSBM}}
\newcommand{\RMHSBM}{\operatorname{RMHSBM}}

\newcommand{\Binom}{\operatorname{Binom}}

\newcommand{\Bern}{\operatorname{Bern}}

\newcommand{\Var}{\operatorname{Var}}
\newcommand{\KL}{\operatorname{KL}}
\newcommand{\TV}{\operatorname{TV}}

\newcommand{\LCA}{\operatorname{LCA}}
\newcommand{\LCAdown}{\LCA \downarrow}

\newcommand{\BIC}{\operatorname{BIC}}
\newcommand{\BIChat}{\widehat{\BIC}}

\newcommand{\internalnodes}{\calI}
\newcommand{\descendants}{\calD}
\newcommand{\leaves}{\calL}
\newcommand{\depth}{\delta}

\newcommand{\distas}[1]{\mathbin{\overset{#1}{\kern\z\sim}}}

\newsavebox{\mybox}\newsavebox{\mysim}
\newcommand{\distras}[1]{%
	\savebox{\mybox}{\hbox{\kern3pt$\scriptstyle#1$\kern3pt}}%
	\savebox{\mysim}{\hbox{$\sim$}}%
	\mathbin{\overset{#1}{\kern\z\resizebox{\wd\mybox}{\ht\mysim}{$\sim$}}}%
}

\newcommand{\argmax}[1]{
\underset{#1}{\text{argmax}}
}

\begin{document}

\title{Testing for Repeated Motifs and Hierarchical Structure in Stochastic Blockmodels}
\author{Al-Fahad Al-Qadhi$^1$}
\author{Keith Levin$^2$}
\author{Vincent Lyzinski$^1$}

\affil{\small $^1$Department of Mathematics, University of Maryland, $^2$Department of Statistics, University of Wisconsin-Madison}

\maketitle

\begin{abstract}
 The rise in complexity of network data in neuroscience, social networks, and protein-protein interaction networks has been accompanied by several efforts to model and understand these data at different scales.
 A key multiscale network modeling technique posits hierarchical structure in the network, and by treating networks as multiple levels of subdivisions with shared statistical properties we can efficiently discover smaller subgraph primitives with manageable complexity. 
 One such example of hierarchical modeling is the Hierarchical Stochastic Block Model, which seeks to model complex networks as being composed of community structures repeated across the network.
 Incorporating repeated structure allows for parameter tying across communities in the SBM, reducing the model complexity compared to the traditional blockmodel. 
In this work, we formulate a framework for testing for repeated motif hierarchical structure in the stochastic blockmodel framework.
We describe a model which naturally expresses networks as a hierarchy of sub-networks with a set of motifs repeating across it, and we demonstrate the practical utility of the test through theoretical analysis and extensive simulation and real data experiments.
\end{abstract}

\section{Introduction} \label{sec:intro}

\newcommand\blfootnote[1]{%
  \begingroup
  \renewcommand\thefootnote{}\footnote{#1}%
  \addtocounter{footnote}{-1}%
  \endgroup
}

Networks, \!\!\!\blfootnote{\!\!\!\!\!\!\!\!\!Code and details can be found at https://github.com/alfahadalqadhi/RMHSBM}
which encode interactions among entities, are ubiquitous in the sciences, arising in fields as diverse as ecology \cite{ecology-fragility,ecology-trophic,primatology}, sociology \cite{BorMehBraLab2009,weakties}, economics \cite{economic-networks-challenges,argentine-banks}, and neuroscience \cite{Sporns2016,Sporns2022}.
As they have become more prevalent, network data has become increasingly complex, and the need to model networks at multiple scales is necessary for efficient modeling and analysis; see, for example, \cite{alqadhi2022subgraphnominationqueryexample}.
A key idea for understanding hierarchical networks is the idea of repeated motif structure \cite{lyzinski15_HSBM}, as this allows for parsimonious model descriptions and enhanced structure discovery in complex network settings.

Effective methods for statistical analysis of network data require that models accurately capture both large-scale network topology and fine-grained network structure \cite{li2022hierarchical}.
As a simple example, consider the social network of a large research university.
Within a department, one is likely to have faculty, staff, and graduate students.
These different sub-populations are likely to interact in ways that are (approximately) repeated across departments.
For example, faculty tend to interact preferentially with staff and other faculty, graduate students interact with specific staff and/or faculty, etc., and we expect similar patterns of this sort to manifest in many different departments across campus.
This suggests a model in which we allow connection probability patterns to be replicated across different parts of a network.
Beyond the presence of repeated high-level structures across a network, we may wish to allow hierarchical structures in our model.

Continuing with our example of modeling a university social network, within a single department, there are groups of students, staff, or faculty who interact differently with one another; e.g., groups of faculty who work on different research sub-areas, or cohorts of students who joined a program in different years. Zooming out to the level of modeling interactions across multiple universities, we may have institutions similar interaction patterns across their departments.
Such patterns may be indicative of certain features.
For example, they may share the role as the largest major in their respective universities, or an interest in the same areas of research.
This presents then two aims: to capture both hierarchical \emph{and} repeated structure in a network.
This motivates the focus of the present manuscript, the {\em repeated motif hierarchical stochastic blockmodel} (RMHSBM).
This model extends the wlel-studied stochastic blockmodel \citep[SBM][]{HolLasLei1983} to incorporate repeated structure into the hierarchical stochastic blockmodel \citep[HSBM][]{lyzinski15_HSBM}.

\subsection{Hierarchical structures in network science}

The presence of hierarchical structure in networks has long been observed and exploited in network science and statistical network inference \citep[see, for example,][]{clauset2006structural, peixoto2014hierarchical,clauset2008hierarchical,li2022hierarchical,gao2023hierarchical}. 
One of the early settings in which hierarchical structure was emphasized is the setting of graph clustering, as the presence of such structure motivates the application of hierarchical or agglomerative clustering methods rather than more classical techniques (e.g., $k$-means).
Early work to estimate the tree structure of the hierarchy in a network was often optimization based \citep[e.g.,][]{clauset2006structural,arenas2008analysis}.
Although these methods tend to be computationally efficient, theoretical optimality is difficult to guarantee.
Nonetheless, these methods are widely used in practice.
For example, \cite{clauset2008hierarchical} estimates the hierarchical structure in a metabolic network, a terrorist association network, and a food chain network.
They also demonstrate the utility of hierarchical structures for edge prediction using the resulting attachment probability parameters. 

Recent work has incorporated hierarchical structure into existing statistical network models, with the stochastic blockmodel and random dot product graph being a natural setting for this \citep[see, for example,][]{peixoto2014hierarchical,li2022hierarchical,amini2024hierarchical}.
In these settings, the HSBM is often posited as a low-rank latent space model, and various spectral methods can be used to approximate the latent vectors.
Then, recursive $k$-means clustering can be applied to retrieve the nested community structure.

One feature of hierarchical structure that we will emphasize in what follows is the role of the tree-structured partition in describing the hierarchy.  
From a generative model perspective, the recent work on $\mathbb{T}$-stochastic graphs \cite{fang2023mathbb} encompasses a wide range of statistical network models as special cases.
In these models, a latent tree structure drives network formation.
The leaves of this tree correspond to the vertices in the observed graph, and the probabilities of edges among these vertices are inversely related to the distances between the leaves on the tree. 
To infer underlying hierarchical tree structure in a network, \cite{li2022hierarchical} and \cite{lei2020consistency} use spectral clustering to partition the network into two subgraphs, and the process is repeated recursively until a stopping criterion is reached.
This partitioning yields a binary tree structure in which the root node represents the full graph, its children represent the two subgraphs obtained in the first partition, and so on.
The first obvious limitation of these methods is the restriction to binary trees, as in each step the graph is partitioned into exactly two subgraphs.
In \cite{li2022hierarchical}, it is further required that the tree is a complete binary tree: if one subgraph triggers the stopping criterion, other subgraphs at the same depth of the tree are not tested for further subdivisions.
\cite{lei2020consistency} do away with this requirement, allowing for a more general class of models.
Their method generalizes beyond binary trees and other symmetry assumptions, as it captures a much more general class of hierarchical graphs.

Often, these methods---and indeed most clustering methods---are posited in the setting of assortative models: they assume that connections to vertices in the same community are more likely than connections to vertices in other community at each level of the tree. 
In the setting of hierarchical structure, this assortativity manifests as follows.
For two nodes (from possibly different communities), the connection probability is inversely related to the distance between their communities on the partition tree. 
While this assortativity assumption is appropriate for many real networks, it need not always be present and methods to tackle disassortative networks (and the continuum between) are needed.
\cite{peixoto2014hierarchical} discusses a method, optimizing a posterior likelihood, which depends not on assortativity but instead on the similarity in the connection probabilities. 
In this work, we largely avoid assumptions around assortativity or disassortativity, as our focus is on assessing the presence or absence of repeated structure in a given partition structure.

\subsection{Motifs and repeated structures}
\label{sec:motif}

Alongside hierarchical structures, networks frequently exhibit repeated structures, often termed {\em motifs}.
The simplest notion of these repeated structures is that of isomorphic subgraphs \citep[see][for discussion and recent advancements on the graph isomorphism problem]{babai2016graph}.
Isomorphism between subgraphs is often too rigid in practice, and the stringent conditions of isomorphism between subgraphs can be relaxed.
One such relaxation can be done using a metric on the similarity between two subgraphs and a threshold for when two subgraphs are considered the similar enough \citep{raymond2002rascal,sussman2019matched,yang2023structural}.
Such a framework is useful in the context of random graphs, as two realizations of the same graph model need not be isomorphic but {\em are} likely to exhibit structural similarities.
The subgraph similarity search problem, analogous to subgraph isomorphism search, searches for subgraphs similar to a given structure. 
Recently, numerous efficient algorithms for motif detection and subgraph comparison have been developed with applications in protein-protein interaction (PPI) networks \citep{yuan2012efficient}, the Wikipedia database \citep{hong2015subgraph}, and the AIDS Antiviral Screen dataset \citep{shang2010similarity}, to name but a few.

In our present hierarchical network setting, motifs can represent repeated structures within the network \cite{lyzinski15_HSBM}.
This repetition can be envisioned as occuring at the generative model level, which allows for repetition of coarse structures (i.e., at the level of communities) while allowing variation at finer scales (i.e., within or among communities).
As an example, in the context of connectomics \citep{Sporns2016,Sporns2022}, while we see symmetry between brain hemispheres at gross anatomical scales, there is significant deviation at the local level \citep{saltoun2023dissociable}. 
In this work, we propose a generative process that can be used to describe different levels of refinement and similarity across levels of a hierarchy, statistical methodology to test for this repetition, and test these methods on simulated and real data.

\section{Background and Setup} \label{sec:setup}

We are concerned in this work with endowing the stochastic blockmodel with a recursive or tree-like structure, in which stochastic blockmodel-like {\em motifs} are repeated in the network.

\subsection{Stochastic Block Models}
\label{sec:sbm}

We define a graph $G=(V,E)$ by a set of vertices $V=[n]$ where $[n]=1,2,\dots,n$ and a collection of edges $E=\{\{i,j\}| i,j\in V\}$. For a set of vertices $V'\subset V$ the induced subgraph by $G$ on $V'$, denoted $G(V')$, is the graph $H=(V',E')$ where $E'=\{\{i,j\}\in E| i,j \in V'\}$.
We encode a network's structure in its adjacency matrix $A \in \R^{n \times n}$.
In the present work, we assume that the observed network is hollow, symmetric, and binary.
That is, the diagonal entries of $A$ are $0$ (i.e., there are no ``self edges''), $A = A^T$ and $A \in \{0,1\}^{n \times n}$.
We note, however, that most of the ideas presented here can be extended to the case of weighted networks and that for the purposes of asymptotic, the on-diagonal entries of $A$ can typically be ignored \citep[see, e.g.,][]{LevAthTanLyzYouPri2017,LevLodLev2022}.

The stochastic blockmodel \citep[SBM;][]{HolLasLei1983} is a widely used statistical model to describe network formation.
Under this model, each of the $n$ vertices belongs to one of the $K$ communities.
This community membership is specified by a vector $\tau \in [K]^n$, with the entries of $\tau$ drawn i.i.d.\ according to a categorical distribution specified by $\pi \in \Delta^{K-1}$, so that $\Pr[ \tau_i = k ] = \pi_k$ for all $k \in [K]$.
Conditional on community memberships $\tau$, the network edges are generated independently, with an edge joining vertices $i,j \in [n]$ with probability $\Pr[ A_{i,j} = 1 \mid \tau ] = B[\tau_i,\tau_j]$, where $B \in [0,1]^{K \times K}$ is a symmetric matrix that encodes the propensity of vertices to form edges based on their community memberships. 

\begin{defn}[Stochastic Blockmodel]
\label{def:sbm}
		We say that an $n$-vertex random graph $G$ with adjacency matrix $A$ is an instantiation of a stochastic blockmodel (SBM) with parameters $(n,K,B,\pi)$, and write $A\sim\SBM(n,K,B,\pi),$ if 
		\begin{itemize}
			\item[i.] The membership probability vector $\pi\in\R^K$ satisfies $\pi_i\geq 0$ for all $i\in[K]$, and $\sum_i\pi_i=~1$;
			\item[ii.] For each vertex $v\in V(G)$, its community membership $\tau_v$ is drawn from $\pi$, independently of all other vertices.
		Given the vector of memberships $\tau \in [K]^n$, we may write $V$ as the disjoint union of $K$ blocks, $V=\calB_1 \sqcup \calB_2 \sqcup \cdots \sqcup \calB_K$.
		\item[iii.] The probability attachment matrix $B\in[0,1]^{K\times K}$ is a symmetric matrix.
		Conditional on the block assignment vector $\tau$, for each pair of vertices $\{i,j\}\in\binom{V}{2}$, $A_{ij}=A_{ji}\sim \Bern(B[\tau_i,\tau_j])$ independently over all $i < j$.
		\end{itemize}
    \end{defn}

In some cases, we will be interested in the behavior of the network conditional on $\tau$ or with $\tau$ chosen deterministically and held fixed.
To handle this, we define a variation of the SBM in which the community assignments are not drawn randomly.
\begin{defn}[Conditional SBM] 
Let $K$ be a positive integer and let $B \in [0,1]^{K \times K}$ be symmetric.
Fix a community membership vector $\tau \in [K]^n$ and generate the upper diagonal entries of a symmetric, hollow adjacency matrix $A$ according to $A_{i,j} \sim \Bern( B[\tau_i,\tau_j] )$, independently over all $i < j$.
We say that the resulting network $A$ is generated according to a {\em conditional SBM} with communication matrix $B$ and community membership vector $\tau$, and write $A \sim \SBM(B, \tau)$.
\end{defn}

\noindent We next proceed to formalize the Hierarchical SBM.

\begin{figure}[t!]
\center
\includegraphics[width=0.6\textwidth]{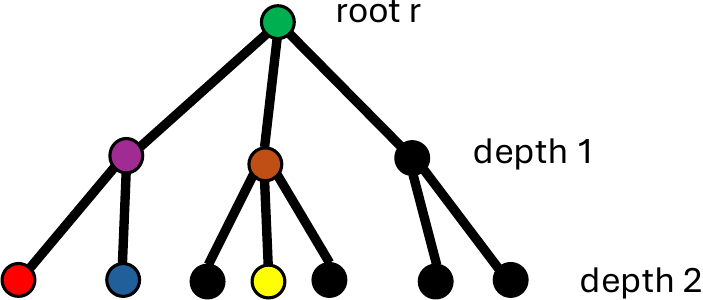}
\caption{An example of a depth 2 rooted tree.  Here (denoting nodes by their color for simplicity) $\LCA$(red,blue)=purple, and $\LCA$(yellow,blue)=green.
We also have that $\LCAdown$(red,blue)$=\{$red,blue$\}$ and $\LCAdown$(yellow,blue)$=\{$purple,orange$\}$.}
\label{fig:tree}
\end{figure}

\subsection{Formalizing the Hierarchical SBM}
\label{sec:hsbmdef}

The hierarchical stochastic blockmodel (HSBM) will allow us to model networks at different levels of resolution.
At one extreme, the Erd\H{o}s-R\'{e}nyi (ER) model, i.e., $\SBM(n,1,p,1)$, has all off-diagonal entries of $A$ (up to symmetry) drawn as independent Bernoulli random variables with the same probability.
Seen as an SBM, this corresponds to all vertices belonging to the same community.
In practice, we often see that there are groups of vertices among which edges have a higher probability than the connections with vertices in other groups, and this idea naturally leads to the stochastic blockmodel. 
Rather than the ER model, which can blur the structure in the adjacency matrix, the SBM allows us to increase the network resolution, at the cost of requiring us to estimate on the order of $K^2$ parameters instead of a single parameter $p$ in the Erd\H{o}s-R\'{e}nyi model.
To balance the competing demands of model expressiveness and model simplicity, we introduce the hierarchical SBM (HSBM).

Formalization of the HSBM has previously been pursued in \cite{lyzinski15_HSBM,li2022hierarchical}.
Here, we adopt the following convention to define the Hierarchical SBM.
Let vertex set $V$ and edge set $E$ be such that $(V,E)$ is a tree.
For a vertex $r \in V$, let $T = ( V, E, r )$ denote the tree rooted at $r \in V$.
Given a rooted tree $T$, we define the following notation:
\begin{itemize}
    \item We let $\leaves(T)$ denote the leaves of $T$ and $\internalnodes(T)$ denote the internal nodes of $T$, so that $V = \leaves(T) \cup \internalnodes(T)$.
\item For any node $t \in V$, we write $\descendants(t) \subseteq V$ to denote the descendants of $t$. We take $t\notin\descendants(t)$ by convention.
\item For a node $t \in V(T)$, we let $\depth(t)$ denote its depth, i.e., the length of the path from the root to $t$. By convention,  $\depth(r)=0$.  If $\depth(v)=k$ for all $v\in\leaves(T)$, then we say the tree is a \emph{depth-k} tree.
\item For two nodes $t_1,t_2$ in tree $T$, we let $\LCA(t_1,t_2)$ denote their lowest common ancestor (i.e., the common ancestor with the largest depth).
\item For two nodes $t_1,t_2$ in tree $T$ we let $\LCAdown(t_1,t_2) \in \binom{ V(T) }{ 2 }$ denote the two children of $\LCA(t_1,t_2)$ encountered on the paths from $\LCA(t_1,t_2)$ to $t_1$ and $t_2$.
\end{itemize}
See Figure~\ref{fig:tree} for an example of this notation in a simple depth-2 tree structure.

Similar to the \emph{$\mathbb{T}$-stochastic graphs} \citep{fang2023mathbb}, the Hierarchical SBM will be defined with an implicit tree structure.
In particular, the tree will encode block memberships at different levels of the hierarchy. 
We call the communities in these different levels {\em metablocks} to emphasize that community structure in this model is defined at mutiple levels of granularity, and we may contain ``communities of communities''.
While the Hierarchical SBM will be defined recursively below, we emphasize the underlying tree structure for clarity.

\begin{defn}[Hierarchical SBM] \label{def:HSBM}
We say that an $n$-vertex random graph $G$ with adjacency matrix $A$ is an instantiation of a $L$-level Hierarchical Stochastic Block Model (HSBM) with parameters 
\begin{itemize}
    \item[-]$n\in\mathbb{Z}^+$, the number of vertices in the graph;
    \item[-]$K\in\mathbb{Z}^+$, the number of metablocks at level-$1$ (i.e., at depth 1 in the tree);
    \item[-]$B\in[0,1]^{K\times K}$ (a symmetric matrix), the level-$1$ metablock-to-metablock communication matrix
    \item[-] $\pi\in\Delta^{K-1}$ (the $K-1$ simplex), the level-$1$ metablock probability assignment vector;
    \item[-] $\{S_i\}_{i=1}^{K}$, a collection of $(L-1)$-level HSBMs
\end{itemize} 
and write $G\sim \HSBM(n,K,B,\pi,\{S_i\}_{i=1}^{K})$ if the following conditions hold:
\begin{itemize}
\item[i.] The vertex set $V=V(G)$ is the disjoint union of $K$ metablocks $V=\calB_1 \sqcup \calB_2 \sqcup \cdots \sqcup \calB_{K}$, where each vertex $v\in V$ is independently assigned to a metablock according to $\pi$.
For each vertex $v\in V(G)$, let $\tau_v$ be the metablock that $v$ is assigned to.
\item[ii.] Conditional on the metablock assignment vector $\tau \in [K]^n$, for each pair of vertices $\{i,j\}$ with $\tau_i\neq \tau_j$, 
$$A_{ij}\stackrel{\text{ind.}}{\sim}\Bern(B[\tau_i,\tau_j]).$$
\item[iii.] For each $i\in[K]$, conditional on the block assignment vector $\tau=(\tau_v)$ and $|\mathcal{B}_i|=n_i>0$,  $S_i$ denotes the model for the $i$-th metablock at level $2$ in the HSBM.
$S_i$, the model for $G[\mathcal{B}_i]$, is a $(L-1)$-level HSBM of the form $\HSBM(n_i,K_{(i)},B_{(i)},\pi_{(i)},\{S_{(i,j)}\}_{j=1}^{K_{(i)}})$, each $S_{(i,j)}$ is itself an $(L-2)$-level $\HSBM(|\mathcal{B}_{i,j}|,K_{(i,j)},B_{(i,j)},\pi_{(i,j)},\{S_{(i,j,k)}\}_{k=1}^{K_{(i,j)} })$, and so on.
Moreover, conditional on the block assignment vector 
		$\tau=(\tau_v)$,
		the collection 
		$\{G\left[\mathcal{B}_k\right]\}_{k=1}^{K}$ are mutually independent.
          \item[iv.] The structure \emph{within} each of the models $\{S_i\}_{i=1}^K$ represents the level $2$ HSBM structure, the structure \emph{within} each of the $\left\{\{S_{(i,j)}\}_{j=1}^{K_{(i)}}\right\}_{i=1}^K$ the level $3$ structure, and so on.  
  At level $L$, the models are indexed by $\vec{v} \in \mathbb{Z}^L$, and each such model $S_{\vec v}$ is an SBM as specified by Definition~\ref{def:sbm}.
  We will denote the block membership vector associated with $S_{\vec v}$ via $\tau_{\vec v}$, so that for a vertex conditioned on being in $u\in\mathcal{B}_{\vec v}$, we have $\tau_{\vec v}(u)=k$ with probability $\pi_{\vec v}(k)$ for each $k$ in $[K_{\vec v}]$.
        \end{itemize}
 \end{defn}
\noindent In the above definition, the notation for denoting subsequent levels of the hierarchy becomes cumbersome quickly.
To remedy this, we use vector notation to subscript the levels of the HSBM.
If at level $k\in[L]$, the metablock model $S_{\vec v}$ for $\vec v$ in $(\mathbb{Z}^+)^k$, represents the $v_k$-th metablock of the $v_{k-1}$-th metablock of the $\cdots$ of the $v_2$-th metablock of the $v_{1}$-th metablock.
This allows for more compact notation when specifying the structure of the HSBM.
For example, at level $3$ in the hierarchy we write $S_{\vec v[1:3]}=S_{(v_1,v_2,v_3)}$ for the model of metablock $\calB_{(v_1,v_2,v_3)}$ with connectivity matrix and $B_{(v_1,v_2,v_3)}$ among the sub-metablocks $\{S_{(v_1,v_2,v_3,i)} \}_{i=1}^{K_{(v_1,v_2,v_3)}}$.

\begin{rem}
\label{rem:htos}
\emph{Note that an SBM according to Definition \ref{def:sbm} can be seen as a $L$-level HSBM (for any $L$) by simply having the level-1 structure encode the SBM blocks, and letting the HSBMs at subsequent levels be 1-block Erd\H os-R\'enyi graphs.  As such, in item iii. of the definition above each $S_i$ could be a standard SBM (according to Definition \ref{def:sbm}) extended to be level $L-1$.}
\end{rem}

\begin{rem}
\emph{In Definition \ref{def:HSBM}, there is an implicit tree structure $T$ defining the hierarchy in the network.
For an $L$-level HSBM, this corresponds to a depth-$L$ rooted tree with $K$ vertices at depth $1$.  These vertices represent the $K$ metablocks $\{S_{(i)}\}_{i=1}^K$.
At depth $2$, there are $\sum_{i=1}^K K_{(i)}$ nodes, with the node representing $S_{(i)}$ at depth $1$ having $K_{(i)}$ offspring in depth $2$.
At depth $3$, there are $\sum_{i=1}^K \sum_{j=1}^{K_{(i)}} K_{(i,j)}$ nodes, with the node representing $S_{(i,j)}$ at depth $2$ having $K_{(i,j)}$ offspring in depth $3$.
We then define the tree inductively to depth $L$.
Given this tree, we define the \emph{traversal function}
$\mathcal{T}:V\rightarrow \mathbb{Z}^{L}$, where  $\mathcal{T}(i)=\vec v$ denotes that vertex $i$ belongs to metablock $S_{\vec v}$ at level $L$ in the hierarchy.
The depth $\delta(t)$ then refers to the index at which the membership vector $\mathcal{T}(i)$ is observed, i.e., $\mathcal{T}(i)_{\delta(t)}$.}
\end{rem}

 \noindent Note that the HSBM is a special case of the SBM, but the hierarhical structure allows us to use fewer parameters to describe the model. For example, the number of parameters for a general SBM with $K>4$ communities is $(K^2-K+2)/2$, meanwhile an \textit{HSBM} with two level-$1$ communities each containing $K/2$ level-$2$ communities would have $(K^2-2K+12)/4$ parameters.
 The hope is that this reduction in complexity facilitates better estimation of the parameters.
 
 We observe that any $L$-level HSBM model can always be written as an SBM, as we now sketch.
Let $G\sim \HSBM(n,K,B,\pi,\{S_{(i)}\}_{i=1}^{K})$.
 Then $G\sim \SBM(n,K^*,B^*,\pi^*)$ where $K^*$ is the total number of possible block assignments at the lowest level, more precisely $K^*=|\calL(T)|$;
 for block $S_{\vec v}$ representing a leaf in the tree, the block membership is given by 
\begin{equation*}
\left(\pi_{\vec{v}(1)}\prod_{i=1}^{L-1}\pi_{\vec{v}[1:i]}(\vec{v}(i+1))\right)\pi_{\vec{v}} .
\end{equation*}
That is, the product of the probabilities along the path from the root of the tree to $S_{\vec v}$.
For two distinct blocks $S_{\vec v}$ and $S_{\vec w}$, for which $\vec w\neq \vec v$, which correspond to distinct leaves in the tree, the probability of an edge between vertices in these two blocks is given by the following:
 Let $i^*(\vec{v},\vec{w})=\min\{i:v(i)\neq w(i)\}$, so that $i^*(\vec{v},\vec{w})-1$ is the level of $\LCA(\vec v,\vec w)$ in the tree.
Write  $i^*=i^*(\vec{v},\vec{w})$ for ease of notation, and if $\vec v=\vec w$, then $i^*=0$ by convention.
We have
 \begin{equation*}
 B_{\vec v,\vec w}=
      B_{\vec{v}[1:(i^*-1)]}[\vec{v}(i^*),\vec{w}(i^*)]=B_{\vec{w}[1:(i^*-1)]}[\vec{v}(i^*),\vec{w}(i^*)],\text{ if }i^*>0.
 \end{equation*}
 If $\vec v=\vec w$, then the connection probabilities among vertices in $S_{\vec v}$ are dictated by the block probability matrix $B_{\vec v}=B_{\vec w}$.

\begin{rem} \label{rem:TSG}
    $\mathbb{T}$-Stochastic graphs, as developed by \cite{fang2023mathbb}, can describe \emph{HSBM}s, but restricted only to graphs that are 
  \begin{itemize}
        \item Weakly Assortative: $B_{\vec{u},\vec{v}}<B_{\vec{w},\vec{v}}$ if $\LCA(\vec{w},\vec{v})\in \LCAdown(\vec{u},\vec{v})$
        \item Assortative: $B_{\vec{u},\vec{v}}<B_{\vec{u}_*,\vec{v}_*}$ if $\delta\left(\LCA(\vec{u},\vec{v})\right)< \delta\left(\LCA(\vec{u}_*,\vec{v}_*)\right)$
        \item Disassortative: $B_{\vec{u},\vec{v}}>B_{\vec{u}_*,\vec{v}_*}$ if $\delta\left(\LCA(\vec{u},\vec{v})\right)< \delta\left(\LCA(\vec{u}_*,\vec{v}_*)\right)$
    \end{itemize}
In particular, this can be achieved using a Bernoulli edge distribution where $B_{\vec{u},\vec{v}}=ce^{-d(\vec{u},\vec{v})}$ for a positive constant $c$, where $d(\vec{u},\vec{v})$ is the length of the shortest path between $\vec{u}$ and $\vec{v}$ on the tree $T$.
Using our notation, let $k=\delta\left(\LCA(\vec{u},\vec{v})\right)$, $L=\delta\left(i\right)$ for $i\in \mathcal{L}$ and $\ell\left(\vec{u}\right)$ be the length of vector $\vec{u}$ then $d(\vec{u},\vec{v})=\ell\left(\vec{u}\right)+\ell\left(\vec{v}\right)-2k=2(L-k)$.
Our definition allows for models that are not weakly assortative, assortative, or disassortative, furthermore, our application in Section \ref{sec:expts} is on a network that does not satisfy any of these conditions. Otherwise, their method would be good for inferring the hierarchical structure.
\end{rem}

\begin{rem} \label{rem:RDPG}
A variation of the HSBM based on the \emph{Random Dot Product Graph} \citep[RDPG;][]{AthFisLevLyzParQinSusTanVogPri2018} can be defined as follows.
Let $\Omega_d\subset\mathbb{R}^d$ where for any $x,y\in \Omega_d$, $x^Ty\in [0,1]$ with a distribution $F$ over $\Omega_d$.
Draw latent vectors $x_{\vec{u}}\stackrel{\text{i.i.d.}}{\sim}F$ for each leaf in the tree $T$, then set $B_{\vec{u},\vec{v}}=(x_{\vec{u}})^Tx_{\vec{v}}$.
\cite{lyzinski2016community} show that the spectral decomposition of the adjacency matrix retrieves the latent vectors associated with the probability matrix $\mathbb{E}(A)$, a process otherwise known as Adjacency Spectral Embedding \citep[ASE][]{SusTanFisPri2012}.
Since the HSBM has repeated vectors in the rows and columns of the probability attachment matrix, the HSBM uses a low-rank (in particular, rank-$K^*$) latent space model.
\end{rem}

\begin{rem}
\label{rem:mxdmmbrship}
 In this paper, we focus on exclusive membership (i.e., at each level, each vertex belongs to one and only one community), but a connection can be made by considering an edge distribution model that accounts for overlapping or mixed memberships.
In essence, from the perspective of edge distributions, an overlapping community is one where the out-of-community connections are consistent for all members but with differing edge distribution for different groups inside the community. What this fails to capture is a further parametric simplification in the case where the out-of-community connection distributions are consistent with the rest of the members of a community for some of the other communities, but differs for others.
In our case, a model is only accepted when that subgraph is put on its own with its own parameters, though one could define an additional parametric adjustment for each subgraph of a given community and discard the parameters that insignificantly differ from 0 in an additional hypothesis test. Details and proofs of the consistency of such a scheme are left as a consideration for future work. \cite{zhang2020detecting} presents a method based on the spectral decomposition of the network adjacency matrix and the $k$-median algorithm, with asymptotic consistency. They also demonstrate favorable performance when compared to other methods, while capturing a very general class of models for overlapping communities.
\end{rem}

Similar to the conditional SBM, sometimes we are interested in the behavior of the HSBM the community structure is known or non-random.
Thus, we define the conditional HSBM entirely analogously to Definition~\ref{def:HSBM}, except that the tree structure $T$ and traversal function $\mathcal{T}$ are known a priori.
We write $$G\sim \HSBM(n,K,B,\{S_i\}_{i=1}^{K}, T, \mathcal{T}).$$
Stated simply, instead of the membership of $\mathcal{B}_{\vec{v}[1:i]}$ being randomly assigned according to probability membership vector $\pi_{\vec{v}[1:i-1]}$, these memberships are given a priori via the traversal function $\mathcal{T}$ which assigns each vertex to its path from root to node in the tree.

\begin{figure}[t!]
\includegraphics[width=1\textwidth]{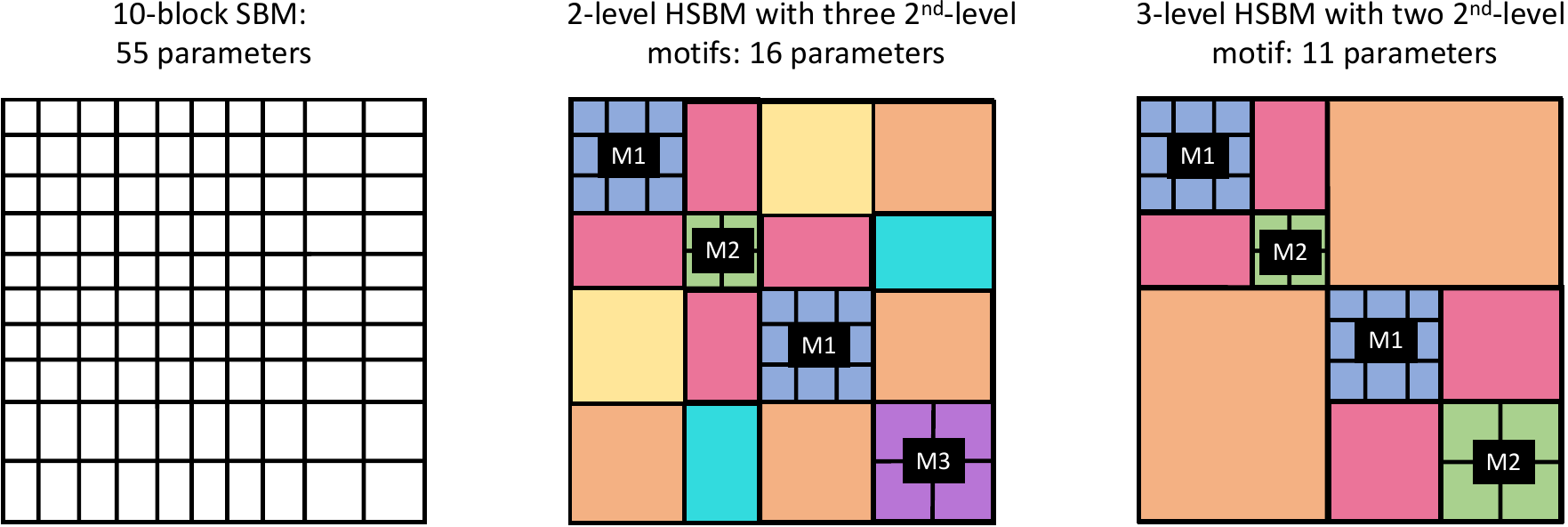}
\caption{Example of refinement of block structure in the HSBM. 
In the center and right figures, each unique color represents a distinct parameter in the model.
The middle (HSBM) model is a refinement of the right (HSBM) model, and the left (SBM) model is a refinement of the middle model.
We can use the BIC formulation below to compare any pair of these three network models.}
\label{fig:hsbm:ex}
\end{figure}

\subsection{Repeated Motif HSBM}
\label{sec:RMHSBM}

The network structure we are after is one where the models for the metablocks (i.e., the $S_{\vec{v}[1:i]}$) in the HSBM are restricted to come from a set of $\Sigma(t)$ possible models.
This is captured below in Definition~\ref{def:motif}, but we first some helpful notation.
Let $G\sim \SBM(n,K,B,\pi)$.
We write $\mathcal{M}(G)=(K,B,\pi)$ to be the \emph{order-independent model} associated with $G$.
Similarly, if $G\sim\HSBM(n,K,B,\pi,\{S_i\}_{i=1}^{K})$,
we write 
\begin{equation*}
\mathcal{M}(G)=(K,B,\pi,\{\mathcal{M}(S_i)\}_{i=1}^{K})
\end{equation*}
to be the \emph{order-independent model} associated with $G$.  
If $G_1$ is an $L^{(1)}$-level HSBM, and $G_2$ is an $L^{(2)}$-level HSBM where, without loss of generality, $L^{(2)}\leq L^{(2)}$), then we say that $G_1$ and $G_2$ are \emph{equivalent models} if, extending $G_2$ to be an $L^{(1)}$-level HSBM as in Remark~\ref{rem:htos}, we have that $\mathcal{M}(G_1)=\mathcal{M}(G_2)$.
Note that if $\mathcal{M}(G_1)=\mathcal{M}(G_2)$, then necessarily for all $\vec{v}\in\mathbb{Z}^{\ell}$ ($\ell\leq L^{(1)}$) we have (where $S^{(1)}_\bullet$ and $S^{(2)}_\bullet$ denote metablock models for $G_1$ and $G_2$, respectively) that $\mathcal{M}(S^{(1)}_{\vec v})=\mathcal{M}(S^{(2)}_{\vec v})$.
\begin{defn}[Repeated Motif HSBM]
\label{def:motif}
Let $G\sim \HSBM(n,K,B,\pi,\{S_i\}_{i=1}^{K})$ be an $L$-level HSBM.
We say that $G$ is a Repeated Motif HSBM (RMHSBM) with motifs $\{\calM_{i}\}_{i=1}^M$, at level-$\tilde L$ (where $\tilde L<L$), written $G\sim \RMHSBM(n,K,B,\pi,\{S_i\}_{i=1}^{K},\tilde L,\{\calM_{i}\}_{i=1}^M)$, 
if the following holds:
\begin{itemize}
    \item[i.] For each metablock model $S_{\vec{v}}$ of $G$ (where $\vec{v}\in \mathbb{Z}^{\tilde{L}}$), there is a $j$ such that $\mathcal{M}(S_{\vec{v}})=\mathcal{M}_j$;
    \item[ii.] If metablock models $S_{\vec{v}}, S_{\vec{w}}$ of $G$, where $\vec{v},\vec{w}\in \mathbb{Z}^{\Ltilde}$ are such that $\vec{v}[1:\tilde{L}-1]=\vec{w}[1:\tilde{L}-1]$, and $\mathcal{M}(S_{\vec{v}})=\mathcal{M}(S_{\vec{w}})$, then for all metablock models $S_{\vec{u}}$ where $\vec{u}\in \mathbb{R}^{\tilde{L}}$ are such that $\vec{v}[1:\tilde{L}-1]=\vec{u}[1:\tilde{L}-1]$, we have that $$B_{\vec{v}[1:\tilde{L}-1]}[w(\tilde{L}-1),u(\tilde{L}-1)]=B_{\vec{u}[1:\tilde{L}-1]}[v(\tilde{L}-1),u(\tilde{L}-1)];$$ i.e., all metablocks with the same parent node as $S_{\vec{v}}, S_{\vec{w}}$ have the same probability of forming an edge to vertices in $S_{\vec{v}},$ and $S_{\vec{w}}$.
\end{itemize}

\end{defn}
\noindent Note that the conditional (conditioning on the traversal structure) HSBM with repeated motifs is defined analogously, and is denoted via (note that the parameter $\mathfrak{m}$ is explained below)
\begin{equation*}
G\sim \RMHSBM(n,K,B,\{S_i\}_{i=1}^{K}, T, \mathcal{T},\tilde L,\mathfrak{m},\{\calM_{i}\}_{i=1}^M).
\end{equation*}
Note that we will assume in the sequel that this conditional structure \emph{explicitly} identifies which motif $\mathcal{M}_i$ each metablock at level $\tilde L$ is equivalent to, and this is represented by the metablock mapping parameter:
\begin{equation*}
\mathfrak{m}:\text{metablocks at level }\tilde L\mapsto \{\calM_{i}\}_{i=1}^M.
\end{equation*}
One of the principal advantages of the repeated motif HSBM framework is that, given the structure of the motifs, the number of parameters for the HSBM model are greatly reduced.  
This is because all substructures equivalent to a given motif share the same block probability parameters and block membership parameters; these need only be estimated once per represented motif instead of separately over each substructure.
See Figure \ref{fig:hsbm:ex} for example.

As discussed above, our definition of the RMHSBM is motivated by the hemisphere-level similarity in connectomes.  
We can model this structure via a two-level, one-motif RMHSBM.
The tree structure for this RMHSBM is displayed below:
\begin{center}
    \label{fig:connectome}
    \includegraphics[width=0.6\textwidth]{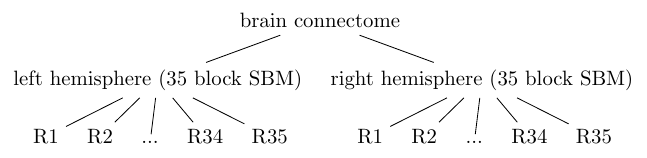}
\end{center}

The tree is two levels, with the first level representing the hemisphere structure of the brain, and the second level the 35 regions of interest (R1, R2, $\ldots$, R35) within each hemisphere, for a total of 70 leaves in the tree.
The meta-communities represented by these regions of interest are modeled as simple Erd\H{o}s-R\'enyi graphs.
The hypothesized motif structural model posits that each hemisphere is an instantiation of a single, twice-repeated, 35 block SBM motif.
We will be comparing \emph{RMHSBM}s with \emph{SBM}s, and to do this sensibly, we need to ensure that the two models are compatible; for that purpose, we give the following definition:
\begin{defn}[$\tau-\calT$ Compatibility]
A node assignment function $\tau:[n]\rightarrow[K]$, a traversal function $\calT$ for a tree $T$ with $K$ leaves are compatible when $\calT(i)=\calT(j)$ if and only if $\tau_i=\tau_j$.
\end{defn}
In the case of the brain connectome, for example, this means that $\tau:[n]\rightarrow[70]$ and for all nodes $v\in Rk$ on the left hemisphere according to $\calT$, then $\tau(v)=k$, and for all nodes $v\in Rk$ on the right hemisphere according to $\calT$, then $\tau(v)=35+k$.

\section{Main Results} \label{sec:results}
Consider the problem of testing whether a graph $G$ comes from a $K^*$ block SBM (with block probability matrix $B^{(1)}$ and known block membership function $\tau$) or from the model 
\begin{equation*}
\RMHSBM(n,K,B,\{S_i\}_{i=1}^{K}, T, \calT,\Ltilde,\mathfrak{m},\{\calM_{i}\}_{i=1}^M),
\end{equation*}
where $T$ has $K^*$ leaves, and the traversal $\calT$ is compatibe with $\tau$. 
From this tree, traversal and motif structure, the RMHSBM model has a (potentially) reduced set of parameters versus the fully general SBM.
For example, all structures in the HSBM corresponding to a single motif share a common set of parameters, and all inter-block SBM connectivity parameters across a layer in the RMHSBM are merged in the RMHSBM as dictated by the flat connectivity across layers in the HSBM.
Let the parameter set given by the tree, traversal and motif of the model be denoted via $\Gamma=\Gamma_{\RMHSBM(T,\calT,\{\calM_{i}\}_{i=1}^M)}$, and for each parameter $\gamma\in \Gamma$, let $\gamma_B$ denote the set of structures in the SBM that are stochastically identical in the RMHSBM and correspond to a single parameter in RMHSBM model space to form $\gamma$.
We remind the reader that the structures merged in $\gamma$ could correspond to a merging of inter-block SBM connectivity structure dictated by the flat connectivity across layers in the HSBM or the merging of structures that are part of a common motif.

We will consider two data settings in which to formulate and test our hypotheses.  
In both settings, we have a population of $N\geq 1$ graphs that serve as our data.
We are essentially testing the goodness-of-fit of the RMHSBM model; tests for SBM goodness-of-fit have been proposed in the literature, see for example \cite{lei2016goodness,karwa2024monte,jin2025network}.
In the first setting, we assume that there is no variation between block parameters in the population, and all elements of the population either follow the $K^*$ block SBM (with fixed $\tau$ and fixed communities $\{\mathcal{B}_i\}$) or the same repeated motif HSBM 
\begin{equation*}
    \RMHSBM(n,K,B,\{S_i\}_{i=1}^{K}, T, \calT,\Ltilde,\mathfrak{m},\{\calM_{i}\}_{i=1}^M)
\end{equation*}
as outlined above.
Testing between the two models can be written as follows, (where $B^{(0)}_\gamma$ denotes the  probability parameter for the structures merged in $\gamma$, and $A$ is the adjacency matrix of a graph in our (assumed i.i.d.) population): 
\begin{equation}
\label{eqn:HSBMnovar}
    \begin{aligned}
        H_0:\, &\text{ For all } \gamma\in \Gamma,\text{ for all } (\ell,k)\in \gamma_B, \text{ and } (v,v')\in\calB_{\ell}\times \calB_k,\ A_{v,v'}\stackrel{ind.}{\sim} \Bern(B^{(0)}_\gamma)\\
        H_1:\, &\text{ For all } (\ell,k)\in[k]^2\text{ and } (v,v')\in \calB_\ell\times \calB_k\text{, } A_{v,v'}\stackrel{ind.}{\sim} \Bern(B^{(1)}_{\ell k}). 
    \end{aligned}
\end{equation}
For every $\gamma\in \Gamma$, let $n_\gamma=\sum_{(l,k)\in \gamma_{B}}n_{lk}$ denote the effective sample size used to estimate $B^{(0)}_\gamma$, and (where $\calB_{\ell,k}=\calB_{\ell}\times \calB_{k}$ if $\ell\neq k$ and $\calB_{\ell,\ell}\binom{\cal{B}_\ell}{2}$ )
    \begin{equation*}
        \Bhat^{(0)}_{\gamma}=\frac{1}{n_{\gamma}}\sum_{(\ell,k)\in \gamma_{B}}\sum_{(v,v')\in \calB_{\ell,k}}A_{v,v'},\quad \text{ and }\quad
    \Bhat^{(1)}_{lk}=\frac{1}{n_{\ell k}}\sum_{(v,v')\in \calB_{\ell,k}}A_{v,v'}.
    \end{equation*}
The log likelihood ratio for the global test of $H_1$ versus $H_0$ in the setting without variation between blocks is then given by
\begin{equation} \label{eq:llr}
-2\lambda_T=2\sum_{\gamma\in\Gamma_T}\sum_{(\ell, k)\in \gamma_B}n_{\ell k}\log\left(\frac{1-B_{\ell k}^{(1)}}{1-B_{\gamma}^{(0)}}\right)+n_{\ell k} \hat B_{\ell k}^{(1)}\log
\left(\frac{ B_{\ell k}^{(1)}(1-B_{\gamma}^{(0)})}{B_{\gamma}^{(0)}(1-B_{\ell k}^{(1)})}\right),
\end{equation}
with maximum likelihood estimator as specified in the following lemma, which is proven in Appendix \ref{sec:lem}.
\begin{lem}
\label{lem:ml}
Under the hypotheses in Equation~\eqref{eqn:HSBMnovar}, the maximum likelihood estimator for $B^{(1)}_{\ell k}$ is $\Bhat^{(1)}_{\ell k}$ and for $B^{(0)}_{\gamma}$ is $\Bhat^{(0)}_{\gamma}$.
\end{lem}
\noindent In this case, a consistent (and extremely powerful; see Section~\ref{sec:expts}) test for the above hypotheses can be constructed using the MLEs and Wilk's theorem.
This test, even at proper level $\alpha$, is too sensitive in practice; note that this (over) sensitivity of global network testing has been noted before, see, for example, the two-sample mesoscale testing work of \cite{macdonald2024mesoscale}.
Indeed, it (rightly!) asymptotically rejects the null for any small fixed deviation across any merged blocks, and is unable to identify cases where there is significant repeated structure within the model.

The global nature of the test also renders it inappropriate in the case where the SBM structures merged in the RMHSBM have small variations in their parameters that are centered around a common mean.
\textcolor{black}{
Indeed, often perfect HSBM structure is not present practically in observed networks, as there will be small deviations within pieces of the hierarchy.
For such a graph, with tree structure given by $\Gamma$, we will define the signal-to-noise ratio of the hierarchical structure via
\begin{equation*}
\mathfrak{s}_{\Gamma}:=\max_{\gamma\in\Gamma}\frac{\frac{1}{n_\gamma}\sum_{\{\ell, k\}\in \gamma_B}\sum_{(v,v')\in\calB_l\times \calB_k} \mathbb{E}(A_{v,v'})}{\sum_{\{\ell, k\}\in \gamma_B}\sum_{(v,v')\in\calB_l\times \calB_k}\mathbb{E}\left( A_{v,v'}-\frac{1}{n_\gamma}\sum_{\{\ell, k\}\in \gamma_B}\sum_{(w,w')\in\calB_l\times \calB_k}\mathbb{E}(A_{w,w'})\right)^2}
\end{equation*}
noting that, as we are assuming block structure, it is implicit that for each $\{\ell,k\}$, $\mathbb{E}(A_{v,v'})$ is constant all $(v,v')\in\calB_l\times \calB_k$. 
For each $\epsilon>0$, we can then test if the HSBM signal strength exceeds a threshold of $\epsilon>0$ by testing the following hypotheses.
}
\begin{equation}
\label{eqn:HSBMvar}
    \begin{aligned}
        H_0^{(\epsilon)} &:
        \mathfrak{s}_\Gamma^{-1}<\epsilon
        ;\\
        H_1 &: \text{ For all } (l,k)\in[k]^2\text{ and } (v,v')\in \calB_l\times \calB_k, E(v,v')\sim \Bern(B^{(1)}_{lk}).  
    \end{aligned}
\end{equation}
Practically, we can realize a model satisfying the signal-to-noise condition $\mathfrak{s}_\Gamma^{-1}<\epsilon$ as follows.  
We can have
that for all $\gamma\in \Gamma$, for all $\{\ell, k\}\in \gamma_B,$  and $(v,v')\in\calB_l \times \calB_k$, we have $\ A_{v,v'}\sim \Bern(B_{\gamma}^{(0)}+s_{\ell,k})$, where the $s_{\ell,k}$'s are suitably small.

Testing in the setting with variations between blocks is useful in our motivating neuroscience application, as it can test whether there is symmetry in the distribution of connectivity within the brain while encompassing individual differences and variations between left and right hemispheres due to environmental effects rather than parametric differences. 

The log-likelihood ratio  in the setting with variations between blocks from \eqref{eqn:HSBMvar} takes the form 
\begin{equation} \label{eq:llrvar} \begin{aligned}
-2\lambda_T 
&=2\sum_{\gamma\in\Gamma_T}\sum_{\{l,k\}\in \gamma_B} 
    \Bigg[ n_{lk}\log \frac{1- B_{lk}^{(1)}}{1-B_{\gamma}^{(0)}-s_{lk}}
+n_{lk} \Bhat_{lk}^{(1)}
    \log \frac{B_{lk}^{(1)}(1- B_{\gamma}^{(0)}-s_{lk})}
        {(B_{\gamma}^{(0)}+s_{lk})(1-B_{lk}^{(1)})} \Bigg] ,
\end{aligned} \end{equation}
The minimizer of this likelihood ratio function is not unique, as there is interdependence between the estimators of $\Bhat^{(0)}_\gamma$ and $\{\hat s_{lk}\}_{{l,k}\in \gamma_b}$ for each $\gamma\in\Gamma_T$. One may use numerical methods with an added condition to find a solution; e.g. $L_p$ regularization.
\subsection{Likelihood Ratio Testing}
\label{sec:LLR}

When the data come from the setting with variations between blocks, a na\"ive use of the likelihood ratio test based on the modeling assumptions in Equation~\eqref{eqn:HSBMnovar} can be shown to be problematic.
This is summarized in the following theorem.
A detailed proof can be found in Appendix \ref{sec:thm4pf}):
\begin{thm} \label{thm4}
Let $\bbP_{0,S}$ indicate the probability under the null hypothesis in Equation~\eqref{eqn:HSBMvar}, where these $\{s_{ij}\}$ further satisfy 
\begin{equation} \label{eq:S}
\epsilon_{S}
=\min_{\gamma\in\Gamma} \min_{(i,j)\in\gamma_B}
    \left|s_{ij}-\frac{1}{|\gamma_B|}\sum_{(\ell,h)\in\gamma_B}s_{\ell h}\right|>0.
\end{equation}
Let 
\begin{align*}
    \lambdahat_{\gamma}=-\sum_{(l,k)\in \gamma_B}n_{lk}\log\left(\frac{1-\Bhat_{lk}^{(1)}}{1-\Bhat_{\gamma}^{(0)}}\right)+n_{lk}\hat B_{lk}^{(1)}\log
\left(\frac{\Bhat_{lk}^{(1)}(1-\Bhat_{\gamma}^{(0)})}{\Bhat_{\gamma}^{(0)}(1-\Bhat_{lk}^{(1)})}\right)
\end{align*}
We then have that
\begin{equation*}
\bbP_{0,S}\left(
\frac{\epsilon_S^2n_\gamma}{9}\leq 
-2\lambdahat_{\gamma}
\leq \frac{8 n_\gamma}{\delta}
\right)
\geq 1-2e^{-2(\min(\epsilon_S/3,\delta/2))^2n_{\gamma}}+2\sum_{(l,k)\in \gamma_B}e^{-2(\min(\epsilon_S/3,\delta/2))^2n_{lk}}
\end{equation*}
\end{thm}

\noindent Note that the assumption on $\epsilon_S$ holds almost surely, as the $s_{ij}$ are assumed continuous, and the number of blocks in the SBM model $K^*$ is assumed fixed. 

From this theorem, we see that the approximation of likelihood ratio testing based on Wilk's theorem is inappropriate.
Wilk's theorem posits that the testing criteria are based on a fixed value that depends only on the difference in the number of parameters or the degrees of freedom.
Theorem \ref{thm4} shows that the LLR statistic here grows as the number of nodes grow. 

Alternatively, if $N>1$, there is hope that averaging over the population will mitigate the effect of the $s_{\ell k}$ terms, and render classical LLR asymptotics again applicable.
As the $s_{\ell k}$ are mean zero, we expect that an average of $N$ such $s$'s to be of order $O(\sqrt{\log(N)/N})$ with high probability, and hence $\epsilon_S$ would be at most of this order as well.
Note however that $N/log(N)$ would need to be of order $n_\gamma$ for the probability lower bound in Theorem \ref{thm4} to not converge to 1, 
and thus $-2\hat{\lambda}_T$ does not grow with $n_\gamma$ with high probability. However, $n_\gamma=\Theta(n^2)$ for a network of size $n$. This makes this requirement prohibitive in practice, as it would require as many networks as there are nodes.

\subsection{Penalized Model Selection and Comparison}
\label{sec:BIC}

In this section, we will consider the problem of choosing between a collection of nested hierarchical SBMs using BIC model selection criteria \citep{schwarz1978estimating, claeskens2008model}.
BIC and its variants have proven to be effective tools for model selection in the SBM setting \citep[see, for example,][]{hu2020corrected,yan2016bayesian,wang2017likelihood}.
By nested models, we mean that the hierarchical structure of the one model is a refinement of the hierarchical structure of another. 
As an example, see Figure \ref{fig:hsbm:ex}, in which case the middle 2-level HSBM model is a refinement of the right 3-level HSBM model, and the left SBM model a refinement of the middle model; in essence, 
model $\mathcal{M}_1$ being a refinement of model $\mathcal{M}_2$ implies that 
\begin{itemize}
    \item[i.] The partition of vertices into blocks provided by $\mathcal{M}_1$ is a refinement for the partition provided by $\mathcal{M}_2$;
    \item[ii.] The set of motifs of $\mathcal{M}_2$ is a subset of the set of motifs in $\mathcal{M}_1$
    \item[iii.] The metablock mapping functions $\mathfrak{m}_1$ and $\mathfrak{m}_2$ are such that if $\mathfrak{m}_1$ maps a metablock to a motif that is present in both models, $\mathfrak{m}_2$ must also map the same metablock to that motif; this is the case with the center panel being a refinement of the right panel in Figure \ref{fig:hsbm:ex}.
\end{itemize} 
This nestedness assumption leads to a similar merging structure as considered in \cite{wang2017likelihood}, and our theoretical results below are of a similar nature.
We can use the BIC formulation below to compare any pair of these three network models, but for ease of exposition, we provide theory for comparing the SBM on the left to one of the HSBM models.

A key assumption below is that we are in the setting of one of the the hypotheses in Equation~\eqref{eqn:HSBMnovar}, so that under the RMHSBM model assumption (i.e., under $H_0$), we observe data from
\begin{equation*}
\RMHSBM(n,K,B,\{S_i\}_{i=1}^{K}, T, \calT,\Ltilde,\mathfrak{m},\{\calM_{i}\}_{i=1}^M)
\end{equation*}
that assumes that the tree structure, traversal, and motif assignments are known a priori. 
Under the SBM model assumption (i.e., under $H_1$),
we are assuming that the block membership function $\tau$ is known a priori and is compatible with $T, \calT,$ and $\mathfrak{m}$.
This avoids the combinatorial unpleasantness (and intractability) that arises when we consider all possible RMHSBM structures possible for a given network. 

We now proceed to show the consistency of the BIC penalized likelihood under these assumptions.
With notation as in Section \ref{sec:results}, the log-likelihood ratio test statistic (with likelihood maximizing plug-in estimators) for the unknown block probabilities can then be written as
\begin{equation} \label{eq:llr}
-2 \lambdahat_T=
2\sum_{\gamma\in\Gamma}
\sum_{\{\ell,k\}\in \gamma_B}n_{\ell,k}
\left[\Bhat_{\ell,k}^{(1)}\log
\left(\frac{\Bhat_{\ell,k}^{(1)}(1-\Bhat_{\gamma}^{(0)})}{\Bhat_{\gamma}^{(0)}(1-\Bhat_{\ell,k}^{(1)})}\right)
+\log\left(\frac{1-\Bhat_{\ell,k}^{(1)}}{1-\Bhat_{\gamma}^{(0)}}\right) \right]
\end{equation}
where $\Phat^{(1)}=\Phat^{(1)}_{\ell,k}$ denotes a $\Bern(\Bhat_{\ell,k}^{(1)})$ measure and 
$\Phat^{(0)}=\Phat^{(0)}_{\gamma}$ denotes a $\Bern(\Bhat_{\gamma}^{(0)})$ measure.
Note that $-2 \lambdahat_T=D_{\KL}(\Phat^{(1)}_{\ell,k}\|\Phat^{(0)}_{\gamma})$ and hence is nonnegative.
In this conditional setting (conditioning on the the tree structure, traversal, and motif assignments being given a priori), 
the standard BIC penalties applied block-wise to the estimated likelihood ratio statistic is defined by, letting $\Lhat_{0,T}$ be the maximum likelihood under $H_0$ and $\Lhat_{1,\tau}$ under $H_1$,
\begin{equation*}
\BIChat_{0,T}=-2\log(\Lhat_{0,T})+|\Gamma|\binom{n}{2}
\quad \text{ and } \quad
\BIChat_{1,\tau}=-2\log(\Lhat_{1,\tau})+ \sum_{\gamma\in\Gamma}|\gamma_B| \binom{n}{2}
\end{equation*}
and hence
\begin{equation}
\Deltahat_{T,\BIC} =\BIChat_{0,T}-\BIChat_{1,\tau}=-2\lambdahat_T-\left(\sum_{\gamma\in\Gamma}(|\gamma_B|-1)\right)\log \binom{n}{2} .
\label{eqn:LLRBIC}
\end{equation}

Below we state a pair of theorems that delineate the asymptotic behavior of $\Deltahat_{T,\BIC}$ under $H_0$ and $H_1$.  
In the theorems below, we make use of the following assumption on the growth rate of $B^{(1)}$: 
\begin{ass}
\label{ass:ass1}
We assume that there exists an $n_0\in\bbZ>0$ and constant $c_1\in(0,1/2)$ such that for all $n>n_0$, we have that 
the following holds for all $\{\ell,k\}$ pairs:
\begin{equation} \label{eq:ass1} 
B_{\ell,k}^{(1)}>\frac{\log n_{\ell,k}}{n_{{\ell,k}}},\quad \text{ and }\quad 
B_{\ell,k}^{(1)}<1/2-c_1.
\end{equation}
\end{ass}

Our first result states that under $H_0$, the penalized BIC difference is negative and the true RMHSBM model is preferred.
A proof can be found in Section~\ref{sec:penLR1} of the Appendix.
\begin{thm} 
\label{thm:penLR1}
With notation as above, given Assumption \ref{ass:ass1}, let 
\begin{equation*}
c_2\leq \frac{1}{16}\left(1-\frac{2|\Gamma|}{(K^*)^2+K^*}
\right)
\end{equation*}
be a constant.
Then we have that for all $n$ sufficiently large, 
\begin{equation*}
\bbP_{0}\left[ \Deltahat_{T,\BIC}<0 \right]
\geq 1 - 
O \left( \exp\left\{ -((K^*)^2+K^*) -\frac{c_2 \log n_{**}}{
2+ 2\sqrt{c_2}/3} \right\} \right)
\end{equation*}
where $n_{**}=\min_{\ell,k} n_{\ell,k}$.
\end{thm}

Let us next consider the case in which the RMHSBM is misspecified.
The next definition quantifies the level of model incompatibility.
Recall that under $H_1$, the model is an SBM$(K^*,B^{(1)},\tau)$ (where $\tau$ is assumed compatible with $\mathcal{T}$).
\begin{defn}
Under $H_1$, we say that the RMHSBM model with block structure given by $\Gamma$ is $(M,\eta)$-incompatible with the hypothesized SBM (from $H_1$)
model if $|\mathbb{E}\Bhat_\gamma^{(0)}-B^{(1)}_{\ell,k}|>\eta$ for at least $M$ triplets $\{\ell,k,\gamma\}$ such that $\{\ell,k\}\in\gamma_B$.
\end{defn}

\noindent Our next result states that under $H_1$ (when the SBM is the true model and not the HSBM), 
the penalized BIC difference is positive and the true SBM model is preferred; we then have
\begin{thm} \label{thm:penLR2}
With notation as above and given Assumption \ref{ass:ass1}, 
assume further that the RMHSBM is $(M,\eta)$-incompatible with the true SBM, where
\begin{equation*}
  \eta^2 M=\omega( (K^*)^2 n_{**}^{-1} \log n.
\end{equation*}
Then for any $c_3>0$, we have that for all $n$ sufficiently large, it holds that  
\begin{equation*}
\bbP_{1}(\Deltahat_{T,\BIC}>0)
\geq 1-4(K^*)^2 \exp\left\{ -\frac{c_3 \log n_{**}}
                        { 2+ 2\sqrt{c_3}/3} \right\},
\end{equation*}
where $n_{**}=\min_{\ell,k}n_{\ell,k}$.
\end{thm}
\noindent Note that a proof of Theorem \ref{thm:penLR2} is given in Section~\ref{sec:penLR2} of the Appendix.

\section{Experiments} \label{sec:expts}

At first, we consider a global log likelihood ratio test.
As we have showed above, if one of the parameters we are testing for similarity fails, we would reject the hypothesis that there is similarity between the communities for sufficiently large $n$. 
This holda since the log likelihood ratio grows with a rate of $O(n_{**}^2)$ where the lead coefficient depends on the parameters. 
In light of this, we propose using conducting multiple hypotheses at the motif level, using the Benjamini-Hochberg correction \cite{benjamini1995controlling} to control False Discovery Rate.
In both cases, we will consider three methods for retrieving the null probabilities: classic likelihood ration testing (implemented via Wilk's theorem), Friedman signed rank for a nonparametric alternative, and ANOVA for testing population versus individual variability.
We next proceed to expound upon these methods.

These different methods consider different levels of granularity in the null hypotheses as well. 
In the case of Wilk's $\chi^2$, we can consider two levels of null granularity.  
More precisely, suppose we have $S$ graphs, $\{A_s\}_{s=1}^S$ with set parameter matrices $\{B_s=[B_{\ell,k;s}]\}_{s=1}^S$; we then consider  
\begin{itemize}
\item[i.] Individual level testing: the null hypotheses can be written (where $n_{\ell,k;s}$ denotes the number of possible edges between community $\ell$ and $k$ in graph $s$),
\begin{align*}
    H_0:
    \text{ For all $\gamma\in \Gamma$ and all $(\ell,k),(\ell',k')\in \gamma_B$, }B_{\ell,k;s}=B_{\ell',k';s}\text{ (global test)}\\
   \left\{H_{0}^{(\gamma)}:
    \text{ For all $(\ell,k),(\ell',k')\in \gamma_B$, }B_{\ell,k;s}=B_{\ell',k';s}\right\}_{\gamma\in\Gamma}\text{ (local tests)} 
\end{align*}
These hypotheses test within each graph at the individual level of the $S$ graphs;
\item[ii.] Aggregated testing:  letting $\tilde N_{\ell,k}=\sum_{s=1}^S  n_{\ell,k;s}$ and $\tilde B_{\ell,k}= \sum_{s=1}^S \frac{n_{l,k;s}}{\tilde N_{\ell,k}}B_{\ell,k;s}$, then the null hypotheses can be written,
\begin{align*}
    H_0:
    \text{ For all $\gamma\in \Gamma$ and all $(\ell,k),(\ell',k')\in \gamma_B$, }\tilde B_{\ell,k}=\tilde B_{\ell',k'}\text{ (global test)}\\
   \left\{H_{0}^{(\gamma)}:
    \text{ For all $(\ell,k),(\ell',k')\in \gamma_B$, }\tilde B_{\ell,k}=\tilde B_{\ell',k'}\right\}_{\gamma\in\Gamma}\text{ (local tests)} 
\end{align*}
Note that this means that the parameters can be different across individuals in the population as long as these weighted sums are not statistically significantly different.
\end{itemize}
 
Meanwhile, the ANOVA test puts an additional condition that the parameters across individuals are the same, replacing the set of parameter matrices above with a single matrix $B$. The nulls can then be written as (note for ANOVA and Friedman, we only test locally; the global tests exhibited the same trend as the $\chi^2$)
\begin{equation*}
\left\{ H_{0}^{(\gamma)}:~\begin{aligned}
    &B_{\ell,k; s}  =B_{\ell,k; s'} \text{ for all } s,s'\in [S], \\
    &B_{\ell,k;s} =B_{\ell',k';s}
    \text{ for all $(\ell,k),(\ell',k')\in \gamma_B$, } \\
&\Var B_s =\sigma^2 I \text{ for all } s\in [S].
\end{aligned} 
\right\}_{\gamma\in\Gamma}
\text{(local tests) }
\end{equation*}
The Friedman test meanwhile, as the data are Bernoulli, tests the same null hypotheses as the ANOVA case with the caveat that independence of the entries of $B_s$ is no longer assumed.

\subsection{Simulations}
\label{sim}
We do the simulations over two models, one to simulate the brain data from \cite{bnu11,bnu12} (obtained via \url{https://fcon_1000.projects.nitrc.org/indi/CoRR/html/bnu_1.html} through \url{https://neurodata.io/mri/}), the other a 3-motif RMHSBM.
Additional details about these models is in Table~\ref{tab:simDetails}. 
In each simulation the parameters are drawn independently from a Dirichlet distribution. 
In each we corrupt the model by changing some of the true parameters to be different from each other at random (injecting different levels of error into the alternative models); these changed parameters are drawn from the same Dirichlet distribution independently of the remainder of the parameters; an example model matrix for this in the BNU1-inspired simulation setting can be found in Figure \ref{ref:corrstencilBNU1} and for the 3-motif, 7-block model in Figure \ref{ref:corrstencil3m}. 
We repeat this twice, first, without any individual differences (as in Eq. \ref{eqn:HSBMnovar}), then by adding a deviations to all parameters to represent individual differences (as in Eq. \ref{eqn:HSBMvar}). 
The small deviations are normally distributed and centered about the true parameters, with a variance dependent on the magnitude of the true parameters, on average these differences have a magnitude of around $1\%$ of the true parameters. We draw the parameters once for all samples. In all cases the parameters had a lower bound cut-off of $0.01$ and an upper bound cut-off of $0.99$ to avoid boundary unpleasantries.

\begin{table}[]
    \centering
    \begin{tabular}{|c|c|c|}
        \hline
        Characteristic & BNU1 Simulation & 3 motif RMHSBM \\
        \hline
        Average nodes per block & 200 & 200\\
        \hline
        Number of blocks & 70 & 70\\
        \hline
        motifs & 1 (repeated twice) & 3 (1 rep. 3 times, 2 \& 3 rep. twice)\\
        \hline
        parameters & 631 & 195\\
        \hline
    \end{tabular}
    \caption{Model specifications for the two RMHSBM models we test.}
    \label{tab:simDetails}
\end{table}

\begin{figure}[h]
    \centering
    \includegraphics[height=0.4\textwidth]{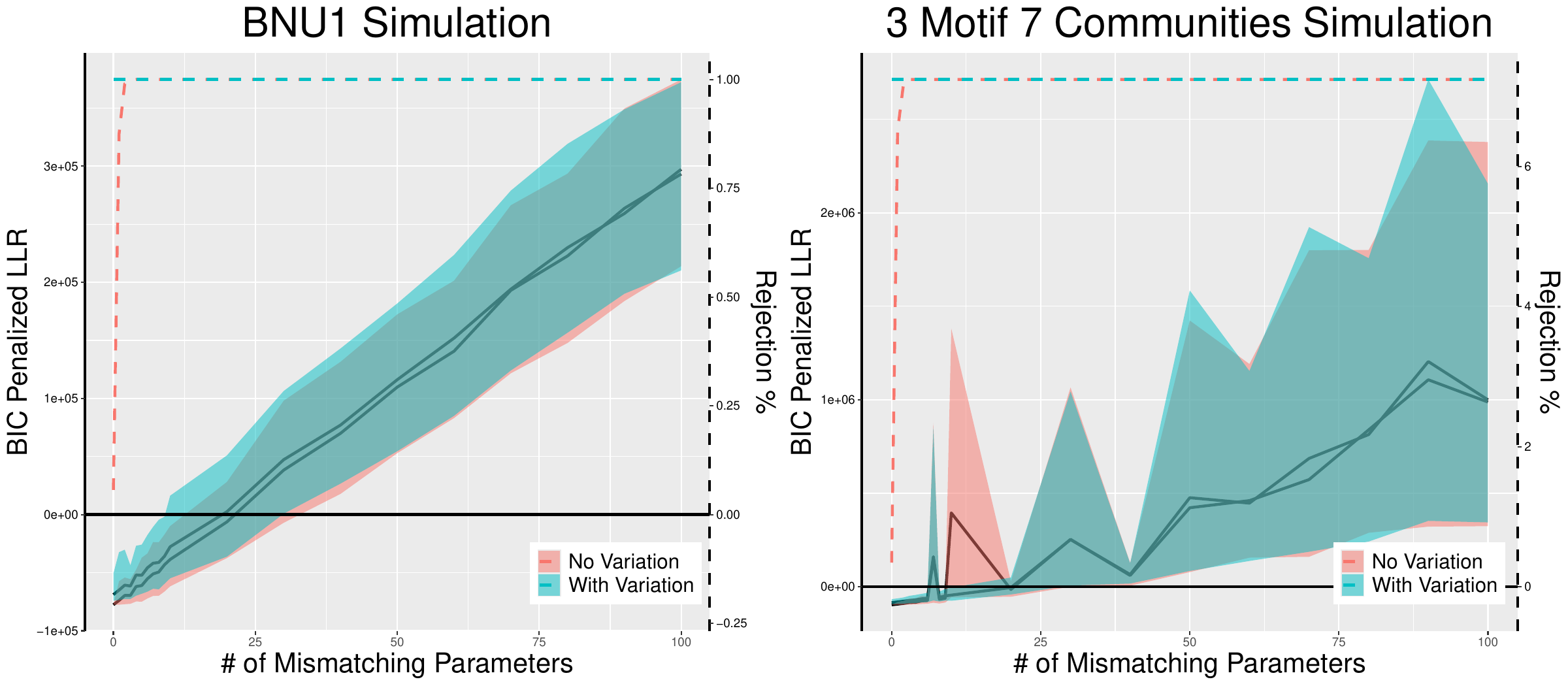} 
    \caption{In each panel: the left axis (and the plots with error bands) shows the penalized log likelihood ratio based on equation \ref{eqn:LLRBIC} while the right axis (and the dashed line plots) show the rejection rate of the global log likelihood ratio test.
    The right panel shows the BNU1-based simulation, the left panel the 3 motif-7 communities simulation.}
    \label{ref:simsLLR}
\end{figure} 

In Figure \ref{ref:simsLLR}, the right axes along with the colored line plots show the average rejection rate over varying numbers of changed parameters for the global $\chi^2$ test of the BNU1 and the 3-motif, 7-communities error-less simulations, respectively; the left axes along with the ribbon plots show the penalized likelihood ratio from equation \ref{eqn:LLRBIC}. 
The results are averaged over 200 simulations (20 for each of the 200 different parameterizations). For each, we plot both the no variation (red) and small variation (blue) settings.
In both panels, we see that in the no variation case this test (correctly) rejects the null hypothesis of repeated structure if we change 10 of the parameters.
In the setting with small variations, the LLR-based test is inconsistent and rejects the null (based upon Wilk's Theorem) even in the case where no parameters are mismatched, as the application of the $\chi^2$ critical value here---as is used in Wilk's theorem---is not appropriate; see Theorem \ref{thm4}.
Meanwhile, we see that in both panels,
the penalized LLR results indicate that the penalized model still selects for repeated structure even when we have around $5\%$ mismatch in the parameters; and that the likelihood ratio test fails to detect similarities in the repeated motif structure even in the case where these few parameters are mismatched. 
This shows the weakness of global testing based on the likelihood ratio, aligning with the theory in which we have shown that these small errors are scaled by a factor of $n^2$ in the test, rendering this test very sensitive to small changes.
\begin{figure}[t!]
    \centering
    \includegraphics[height=0.5\textwidth]{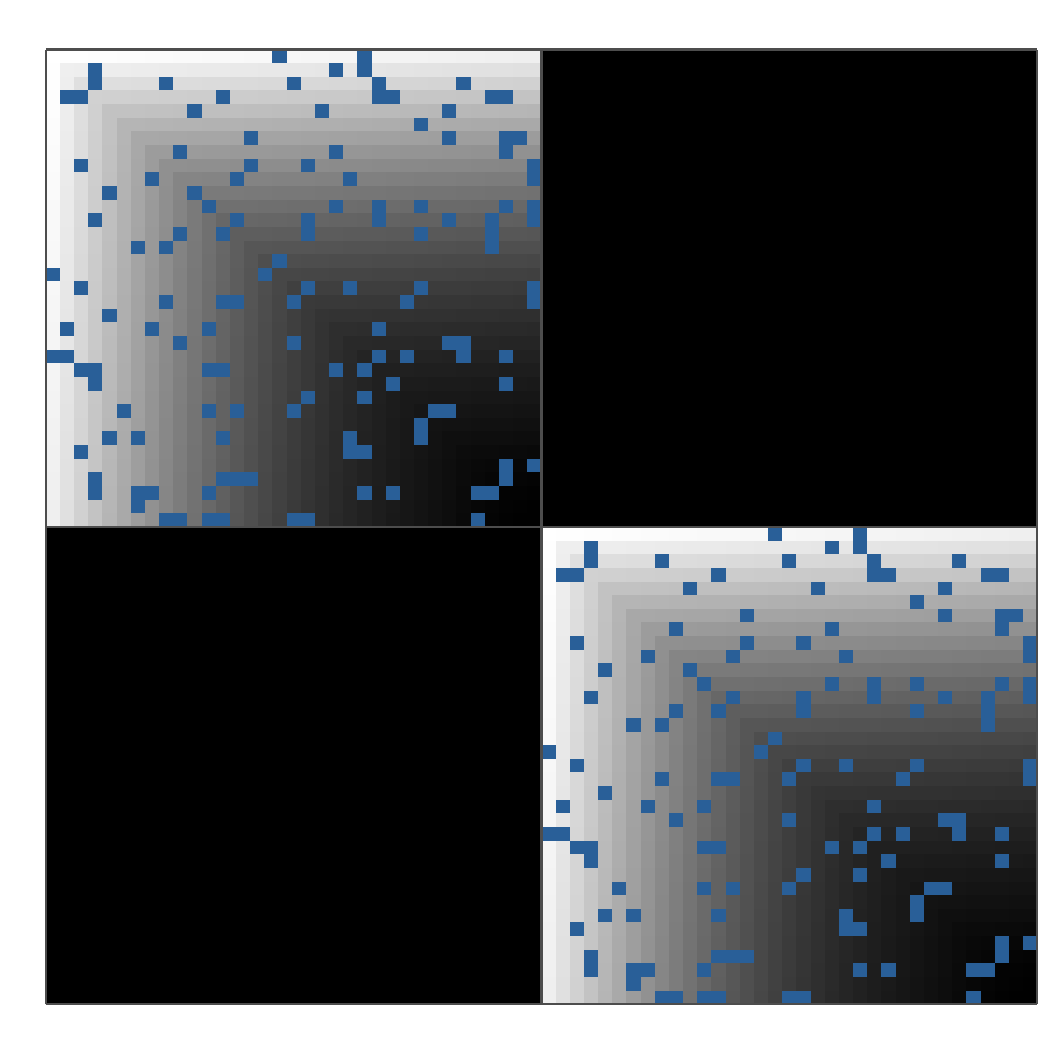}
    \caption{Example model matrices for the BNU1 simulation, 70 blocks, where the first 35 are similar to the second 35 with one cross community connection parameter between them. Blue squares indicate the parameters we set to be different from their supposed pair, while the remaining shades of grey indicate similarity.}
    \label{ref:corrstencilBNU1}
\end{figure}
\begin{figure}[t!]
    \centering
    \includegraphics[height=0.5\textwidth]{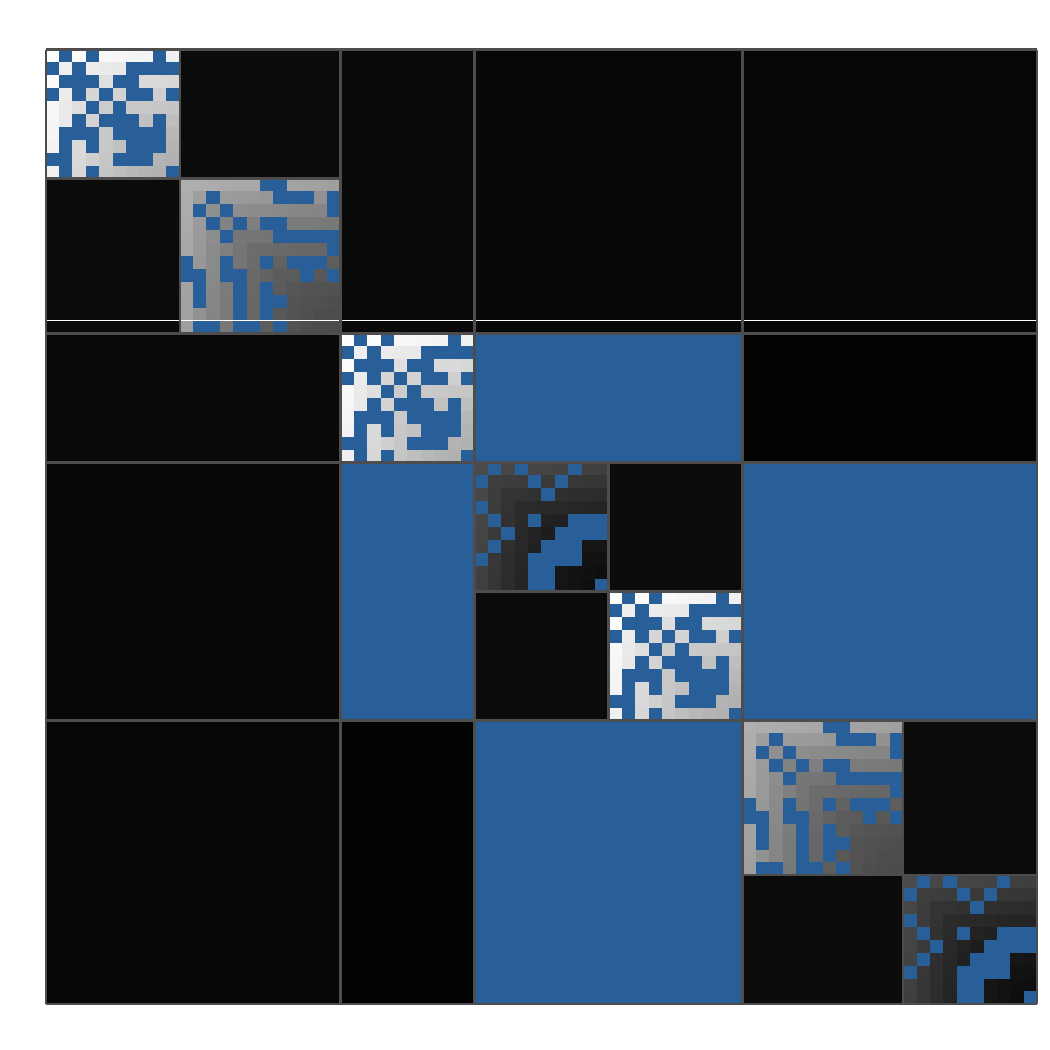}
    \caption{Example model matrix for the 3-motif, 7-communities simulation, 70 blocks, where the block groups $(1,3,6), (2,6), (4,7)$ share a motif. Blue squares indicate the parameters we set to be different from their supposed pair, while the remaining shades of grey indicate similarity.}
    \label{ref:corrstencil3m}
\end{figure}

\begin{figure}[t!]
    \centering
    \includegraphics[height=0.5\textwidth]{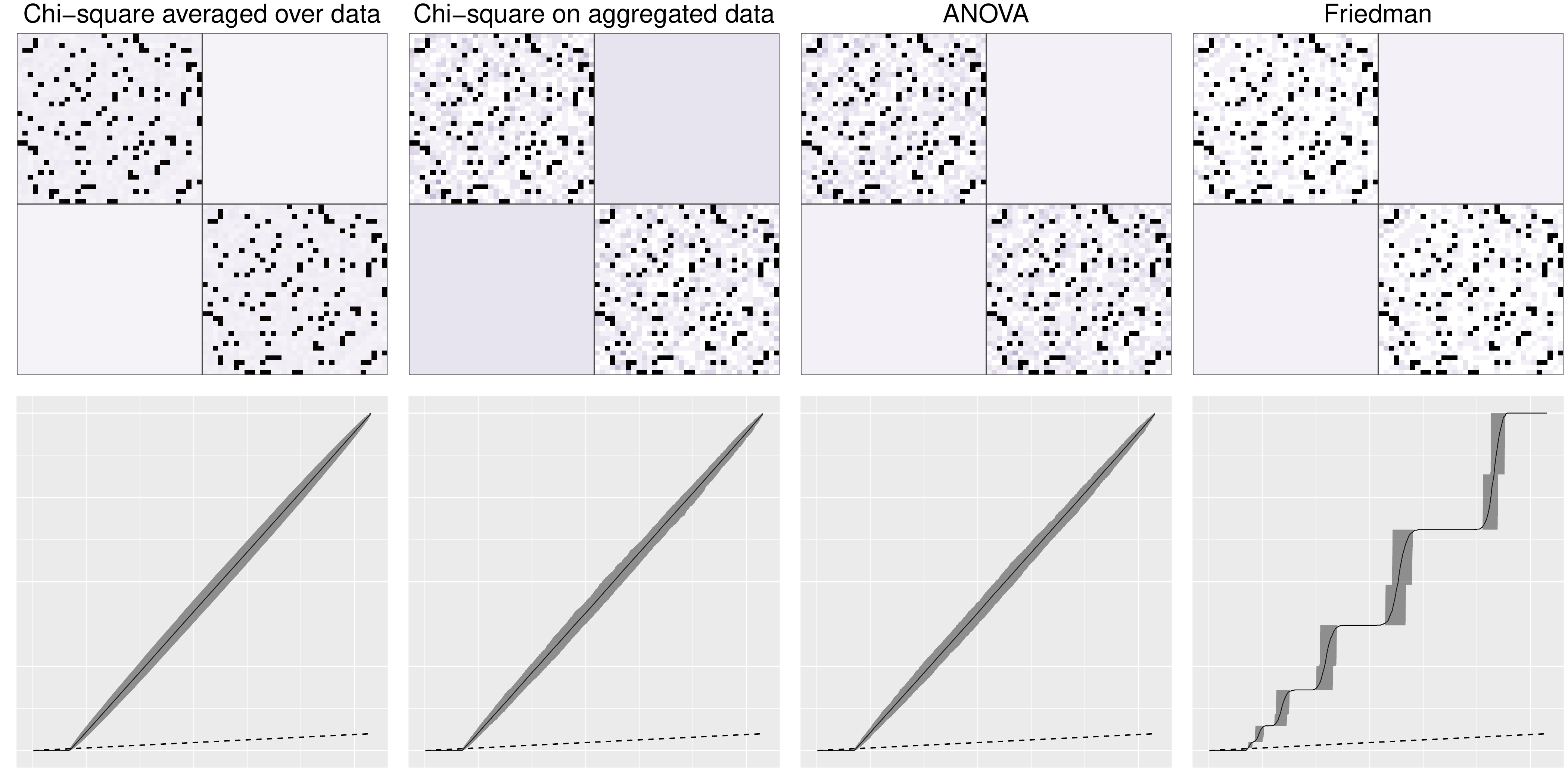}
    \caption{On the top row, are the rejection matrices for (from left to right)  the averaged individual $\chi^2$ tests, the aggregated $\chi^2$ test, ANOVA test, and Friedman's signed rank test on the error-less BNU1 simulation.In black are the correctly rejected blocks, in blue are the incorrectly rejected blocks (Type II error), in white are the blocks that correctly failed to reject, and in red are the blocks are incorrectly failed to reject (Type I error). On the bottom row are the P-value profile for each method respectively, the solid line represents the P values of each block in increasing order, the dashed line represents the Benjamini-Hochberg rejection line. Blue colors are exaggerated to make the comparison clear, all models similarly well.}
    \label{ref:BNU1mat}
\end{figure}

The next approach is to test individual motifs separately, with a Benjamini-Hochberg correction, using the methods discussed above. 
In Figure \ref{ref:BNU1mat}, we see the results of the averaged $\chi^2$ individual testing, the aggregated $\chi^2$ test, ANOVA, and Friedman's signed rank test from left to right, respectively.
Note that this figure considers the model that does not account for variations in the parameters across individuals. 
In the first row are the rejection matrices showing correctly rejected blocks in black, incorrectly rejected blocks in blue, and incorrectly accepted blocks in red;
note that the aggregated methods are more resilient to incorrectly accepting blocks, as indicated by the absence of red in this example. Meanwhile, individual testing is more susceptible to these errors. 
Moreover, the aggregated tests are susceptible to incorrect rejection of blocks, as indicated by the blue squares.
The averaged individual testing, on the other hand, is not free of these errors, but with the correct threshold one may circumvent these errors; the method for selecting this threshold is not clear however. Keep in mind that the darkness of the blue squares is not linearly indicative of the rejection rate; instead it remapped for better visibility. On the second row, we see the profile of the p-values in increasing order; we refer to this as the p-profile going forward. The p-profiles are all linear, indicating a good fit for these tests. In Figure \ref{ref:3mmat}, we see similar results when testing for the 3-motif, 7-communities model.
\begin{figure}[t!]
    \centering
    \includegraphics[height=0.5\textwidth]{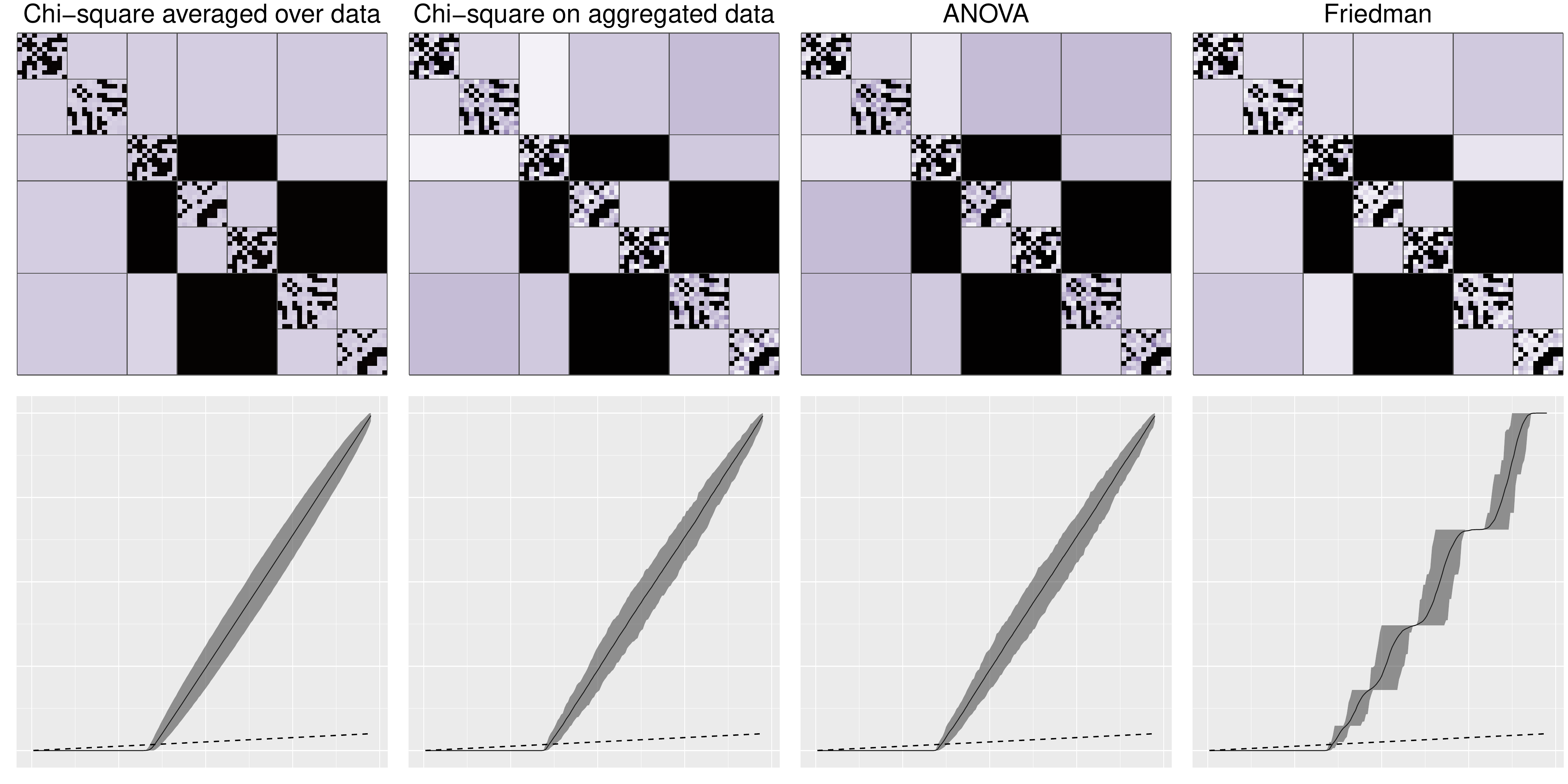}
    \caption{On the top row, are the rejection matrices for (from left to right)  the averaged individual $\chi^2$ tests, the aggregated $\chi^2$ test, ANOVA test, and Friedman's signed rank test on the error-less 3-motif, 7-communities simulation. In black are the correctly rejected blocks, in blue are the incorrectly rejected blocks (Type II error), in white are the blocks that correctly failed to reject, and in red are the blocks are incorrectly failed to reject (Type I error). On the bottom row are the P-value profile for each method respectively, the solid line represents the P values of each block in increasing order, the dashed line represents the Benjamini-Hochberg rejection line. Blue colors are exaggerated to make the comparison clear, all models similarly well.}
    \label{ref:3mmat}
\end{figure}

Moving on to modeling small variations in the sample, we see a different picture. In Figures \ref{ref:BNU1matcorr} and \ref{ref:3mmatcorr}, we see the results of adding small variations to the BNU1 and 3-motif, 7-communities models, respectively. In the presence of small variations, the flaws in both $\chi^2$ tests become apparent. False rejections become very common as indicated by the large amount of blue squares, compared to the ANOVA and Friedman's signed rank tests, respectively. The p-profiles of the $\chi^2$ are not linear indicating that the modeling assumptions are not a good fit in this case. ANOVA and Friedman signed rank tests have a linear p-profile indicating the resilience of their modeling assumptions to small variations. These small variations simulate individual differences across the population.

\begin{figure}[t!]
    \centering
    \includegraphics[height=0.5\textwidth]{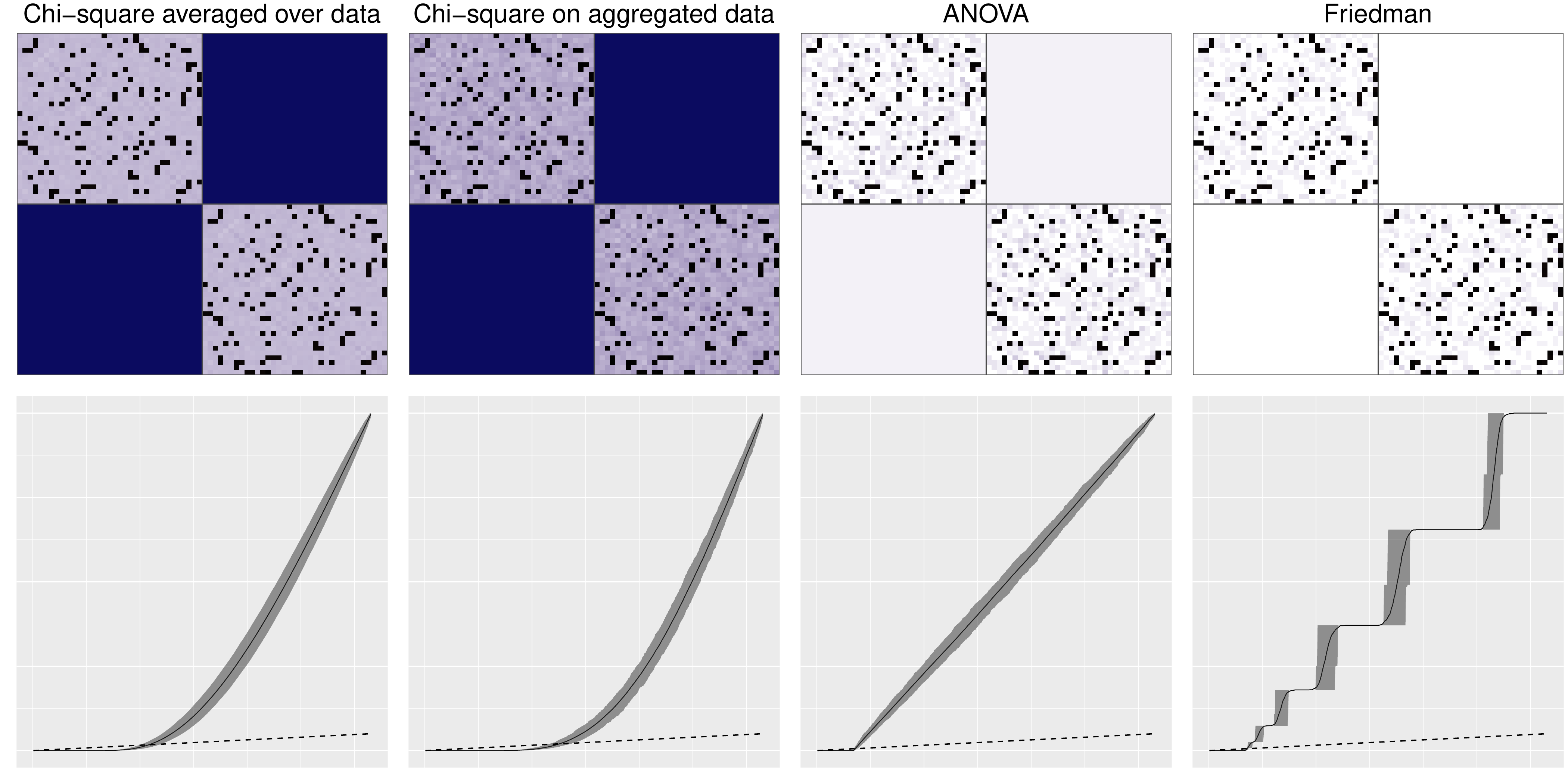}
    \caption{On the top row, are the rejection matrices for (from left to right)  the averaged individual $\chi^2$ tests, the aggregated $\chi^2$ test, ANOVA test, and Friedman's signed rank test on the error-full BNU1 simulation.In black are the correctly rejected blocks, in blue are the incorrectly rejected blocks (Type II error), in white are the blocks that correctly failed to reject, and in red are the blocks are incorrectly failed to reject (Type I error). On the bottom row are the P-value profile for each method respectively, the solid line represents the P values of each block in increasing order, the dashed line represents the Benjamini-Hochberg rejection line.}
    \label{ref:BNU1matcorr}
\end{figure}

\begin{figure}[h]
    \centering
    \includegraphics[height=0.5\textwidth]{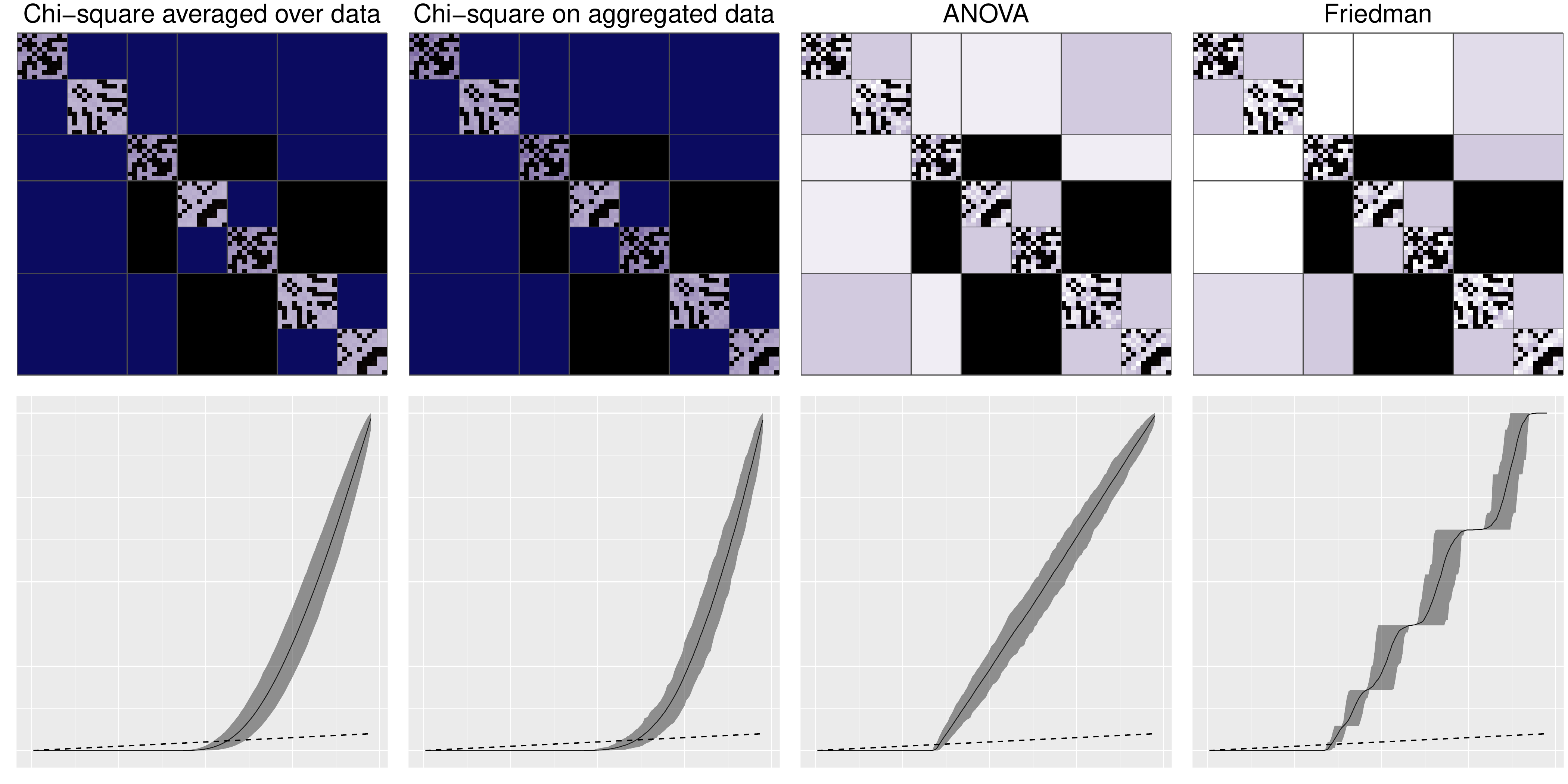}
    \caption{On the top row, are the rejection matrices for (from left to right)  the averaged individual $\chi^2$ tests, the aggregated $\chi^2$ test, ANOVA test, and Friedman's signed rank test on the error-full 3-motif, 7-communities simulation. In black are the correctly rejected blocks, in blue are the incorrectly rejected blocks (Type II error), in white are the blocks that correctly failed to reject, and in red are the blocks are incorrectly failed to reject (Type I error). On the bottom row are the P-value profile for each method respectively, the solid line represents the P values of each block in increasing order, the dashed line represents the Benjamini-Hochberg rejection line.}
    \label{ref:3mmatcorr}
\end{figure}

\subsection{Application: Testing for Community Structure in The Brain.}

Identifying the regions of the brain that share a distributional similarity can help identify regions that have similar functionality or those that are independent. This analysis is performed on the BNU1 data set from \citep{bnu11,bnu12}, a data set of 57 fMRI scans of the human brain of adults at rest. After dividing the brain into spatial regions commonly referred to as voxels, fMRI scans measure changes in blood flows. These changes are interpreted as connections between the voxels that serve as the nodes, producing a graph. The data classifies these nodes into regions that represent larger spatial portions of the brain, forming a natural block structure. Additionally, the voxels are classified based on which half of the brain they are located. 
Combining the classifications, regions and halves gives rise to a natural hierarchical structure for the brain. Table \ref{tab:graphDetails} describes the details of these graphs.

\begin{table}[t!]
    \centering
    \begin{tabular}{|c|c|}
        \hline
        Characteristic & Range \\
        \hline
        Number of vertices & 9312-15631\\
        \hline
        Number of connections & 204022-539259\\
        \hline
        Average degree & 43-74\\
        \hline
        Number of blocks & 70\\
        \hline
        Average block size & 149-211\\
        \hline
    \end{tabular}
    \caption{The range of number of vertices, connections, average degree, and block sizes of the 57 individuals from the BNU1 dataset.}
    \label{tab:graphDetails}
\end{table}

Using the posited hierarchical structure, we perform different analysis techniques to describe the possibility of repeated structure across hemispheres and to understand the potential reduction in the number of parameters needed to describe the distribution of these networks. More precisely, we attempt to answer the questions of whether there is symmetry in the edge distribution across the two halves of the brain.
In Figure \ref{ref:BNU1dataresults} we plot the result of our 4 procedures, in black are the rejected blocks, in white are the blocks that failed to reject, and in red are the blocks that are identically zero in both hemispheres and which trivially fail to reject.
In the figure, in black are the rejected blocks, in white are the blocks that failed to reject, and in red are the blocks that are identically zero in both hemispheres which always fail to reject.
On the left, the results for Wilk's $\chi^2$-test with Benjamini-Hochberg simultaneous testing are displayed. We average over all the individuals since each individual is tested on their own. We note that many of the similarities we fail to reject are instances where both the left and right block probabilities are identically zero on both hemispheres, as indicated by the red colored blocks. From the simulations, we learned that it is difficult to draw any conclusions about the partially rejected tests. 

Next, we aggregated the samples and then performed the chi-square test. The results of these tests are displayed in Figure~\ref{ref:BNU1dataresults} in the second column from the left. Once again, a large part of the tests we fail to reject correspond to parameters that are identically $0$. In this case, we reject most of the ambiguous blocks from the averaged individual tests. This is in stark contrast with the ANOVA test, as we see in Figure~\ref{ref:BNU1dataresults} in the second to last column. Here, the model suggests that many of the parameters have a common mean with differences that are normally distributed. 

\begin{figure}[h]
    \centering
    \includegraphics[height=0.5\textwidth]{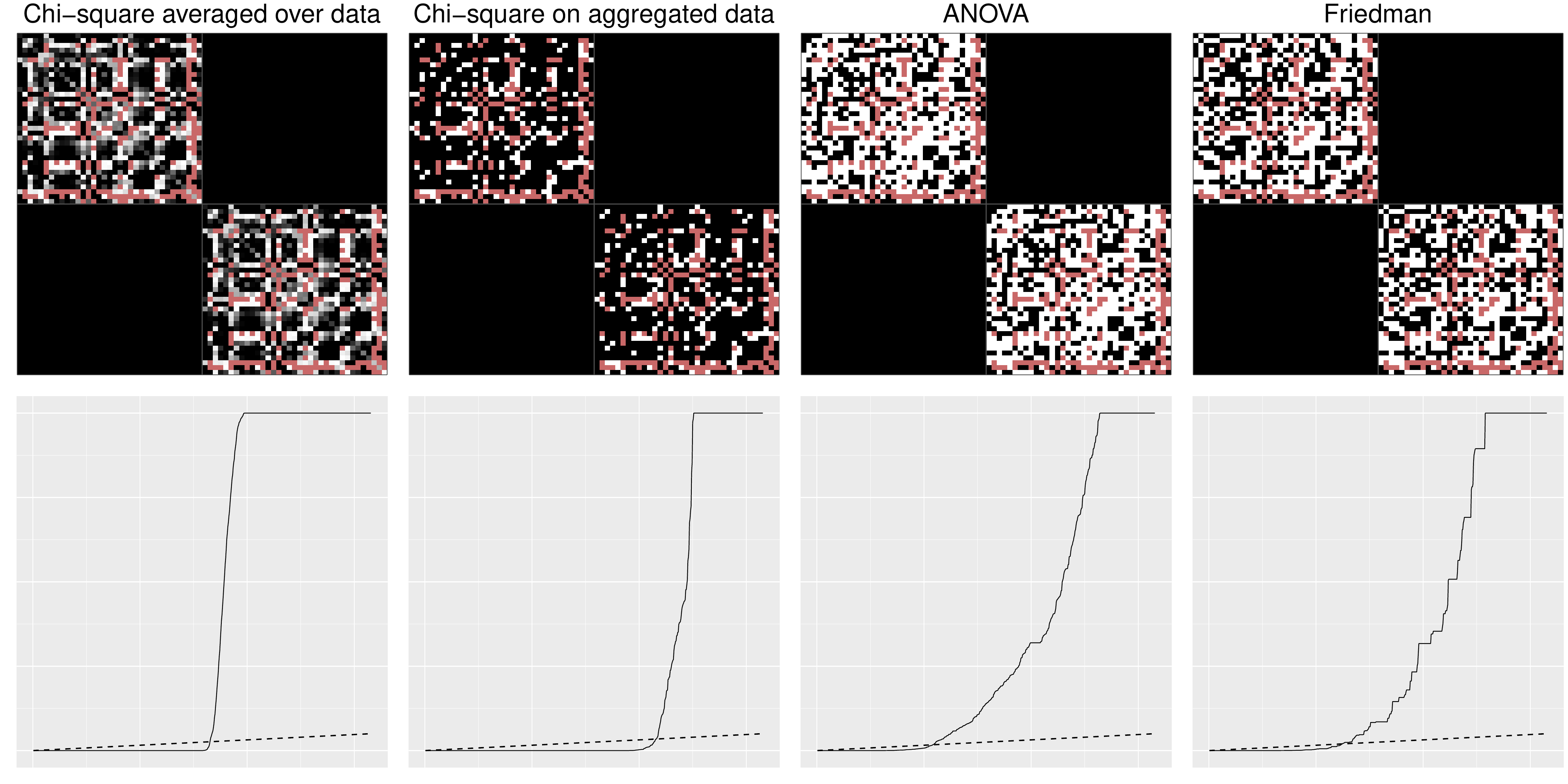}
    \caption{On the top row, are the rejection matrices for (from left to right)  the averaged individual $\chi^2$ tests, the aggregated $\chi^2$ test, the ANOVA test, and Friedman's signed rank test on the BNU1 dataset. In black are the rejected blocks, in white are the blocks that failed to reject, and in red are the blocks that are identically zero in both hemispheres which always fail to reject. On the bottom row are the P-value profile for each method respectively, the solid line represents the P values of each block in increasing order, the dashed line represents the Benjamini-Hochberg rejection line.}
    \label{ref:BNU1dataresults}
\end{figure}

 Finally, we pair the parameters on the left and right hemispheres for each individual, then perform a Friedman signed-rank test.
 This is tested against a mean of zero, and the rejection of the null hypothesis would suggest that the pairs are not exchangeable. 
 This implies (for instance) that the parameters on the left and right hemispheres do not have the same distribution. 
 The results of this test are shown in figure \ref{ref:BNU1dataresults} in the last column. These results largely agree with ANOVA, indicating statistically significant similarity between the left and right hemispheres with small variability across the population, with results similar to those in the simulation setting of Figure \ref{ref:BNU1matcorr}.

\section{Discussion and Conclusion} \label{sec:conclusion}
Structural similarity in graphs can be observed in many disciplines; RMHSBM models provide a general framework for describing graphs with similarity defined in a hierarchy. 
Such models can drastically reduce the number of parameters required to describe these networks which is especially important for large networks with many small communities.
Such simplification can help us better understand and estimate models for these networks.
In this paper, we used theory and simulations to show the difficulty of global testing in situations where there is partial similarity in the hierarchy.
In this case, the sensitivity of likelihood ratio tests to small differences overwhelms the signal in a potentially repeated structure (as shown in the penalized LLR setting).
Instead, we propose locally more robust methods that work to identify a repeated structure at the motif level. 

In addition, we applied these methods to the BNU1 data set and found some structural similarity between the left and right hemispheres in human brains.
For naturally developing networks, confounding environmental effects can cause small perturbations that challenge the underlying symmetry.
By employing a hierarchical structure, we can study these networks at multiple levels of resolution to uncover the underlying symmetry.
However, even with hierarchical thinking, global testing suffers on two fronts: perturbations may lead to false rejections and failure to discover symmetries when some but not all sub-communities admit it.

We give theoretical support that likelihood ratio testing in the presence of small perturbations (that grow at a rate $\omega\left(\frac{1}{\sqrt{n}}\right)$) rejects any similarity under classical LLR testing assumptions.
From simulations, we see that both ANOVA and Friedman's signed rank tests are resilient to these perturbations, discovering the underlying symmetry in a setting where likelihood-ratio testing fails.
In light of these discoveries, we apply these methods to test for symmetries between the hemispheres of human brains. We discover that when we account for these perturbations by using ANOVA or Friedman's test, there are some symmetries at the mesoscale.
Also, during our work with brain imaging data, we discovered that functional regions of the brain do not admit the associative structure defined in \ref{rem:TSG} on which most clustering algorithms rely.

We used the natural clustering given by the data set, but a method that can reliably retrieve these structures and find further subdivisions of them to retrieve a hierarchical structure would allow for this testing to be performed on multiple levels to explore the degree to which there is symmetry. 
An extension of this framework to mixed memberships can be considered here; if the attachment probability in a community is lower than another and a node has mixed membership, it will naturally lead to the break from associativity conditions.
In addition to symmetry, we discussed how repeated structure is related to subgraph matching.
In this work, the matches were already known, and we tested whether they represent a repeated motif.
Further study of the extent to which a repeated motif can lead to the discovery of matches would be another direction to extend this work.

\vspace{2mm}

\noindent\textbf{Acknowledgments:}  
Support for this research was provided by the University of Wisconsin-Madison, Office of the Vice Chancellor for Research and Graduate Education with funding from the Wisconsin Alumni Research Foundation.
The authors gratefully acknowledge support from the JHU HLTCOE.

\bibliographystyle{plainnat}
\bibliography{biblio}

\newpage
\onecolumn
\appendix
\section*{Appendix} \label{sec:apx}

Below we collect technical lemmas and proofs related to our main results.

\section{Supporting results}
\label{sec:supp}

We will make use here of the following theorem, adapted here for our present purposes.
(Theorem 3.3 in \cite{chung2006concentration})

\begin{theorem} \label{thm:chung}
Let $X\sim \Binom(n,p)$, and let $t>0$.
We then have
\begin{equation*} \begin{aligned}
\bbP\left(\frac{X}{n}-p\leq -t\right)
&\leq \exp\left\{-\frac{t^2n}{2p} \right\}\\
\bbP\left(\frac{X}{n}-p\geq t\right)
&\leq \exp\left\{-\frac{t^2n}{2p+2t/3} \right\} .
\end{aligned} \end{equation*}
\end{theorem}

We will also make use of the following two Pinsker-like inequalities that relate the Kullback-Liebler divergence to the total variation distance.
To wit, let $P$ denote a Bernoulli distribution with probability $p$ and $Q$ a Bernoulli distribution with probability $q$.
Without loss of generality, we let $q<1/2$.
Let $D_{\KL}(P\|Q)$ denote the Kullback-Leibler divergence between $P$ and $Q$ and $d_{\TV}(P,Q)$ the total variation distance between the measures.
Then
\begin{equation} \label{eq:pinsker} 
2 d_{TV}(P,Q)^2=2(p-q)^2
\le D_{KL}(P\|Q)
\le \frac{2}{\min\{q,1-q\}} d_{TV}(P,Q)^2=\frac{2}{q}(p-q)^2.
\end{equation}
See Lemma 2.5 in \cite{tsybakov} for the lower-bound, and Lemma 4.1 in \cite{gotze2019higher} for the upper bound.

\section{Proof of Lemma \ref{lem:ml}}
\label{sec:lem}
\begin{proof}[Proof of Lemma \ref{lem:ml}]
Under the null:
\[L(A;B^{0}|\tau)=\prod_{i<j}(B^{0}_{\tau_i\tau_j})^{A_{ij}}(1-B^{0}_{\tau_i\tau_j})^{1-A_{ij}}\]
Defining the following function
\begin{equation*}
\delta_\gamma(i,j)=\begin{cases}
1 &\mbox{ if } (i,j)\in \mathcal{B}_{l}\times\mathcal{B}_{k} \mbox{ for } (l,k)\in \gamma_B\\
0 &\mbox{ otherwise }
\end{cases}
\end{equation*}
and defining $n_\gamma=\sum_{i<j}\delta_\gamma(i,j)$, we can write,
\begin{align*}
    \ell(B^{0})=&\log\left(L(A;B^{0}|\tau)\right)\\
    =&\sum_{i<j}A_{ij}\log(B^{0}_{\tau_i\tau_j})+(1-A_{ij})\log(1-B^{0}_{\tau_i\tau_j})\\
    =&\sum_{\gamma\in \Gamma}\sum_{i<j}A_{ij}\log(B^{0}_{\gamma})+(1-A_{ij})\log(1-B^{0}_{\gamma})\\
    =&\sum_{\gamma\in \Gamma}n_\gamma\hat{B}^{(0)}_{\gamma}\log(B^{0}_{\gamma})+(1-\hat{B}^{(0)}_{\gamma})n_\gamma\log(1-B^{0}_{\gamma})
\end{align*}
Where $\hat B^{0}_{\gamma}=\frac{1}{n_\gamma}\sum_{i<j}A_{ij}\delta_\gamma(i,j)$. 
Since, $\log$ is a strictly increasing function on the support of $L(A;B^{0}|\tau)$, we get
\begin{equation*}
\argmax{B^{0}\in (0,1)^{K\times K}}\{\ell(B^{0})\}
= \argmax{B^{0}\in (0,1)^{K\times K}}\{L(A;B^{0}|\tau)\}
\end{equation*}
We note that since $B^{0}_{\gamma} \in (0,1)$, we have
\begin{equation*}
\frac{\partial \ell(B^{0})}{\partial B^{0}_{\gamma}} = 0
\end{equation*}
if and only if
\begin{equation*}
B^{0}_{\gamma}
=\hat{B}^{(0)}_{\gamma}.
\end{equation*}
Writing
\begin{equation*}
    \ell(B^{0})=\sum_{\gamma\in \Gamma_T}n_\gamma\hat{B}^{(0)}_{\gamma}\log(B^{0}_{\gamma})+(1-\hat{B}^{(0)}_{\gamma})n_\gamma\log(1-B^{0}_{\gamma}),
\end{equation*}
we have
\begin{align*}
    &\max_{B^{0}\in (0,1)^{K\times K}} \ell(B^{0})\\
    &~~~~~~=\max_{B^{0}}\left\lbrace\sum_{\gamma\in \Gamma_T}n_\gamma\hat{B}^{(0)}_{\gamma}\log(B^{0}_{\gamma})+(1-\hat{B}^{(0)}_{\gamma})n_\gamma\log(1-B^{0}_{\gamma})\right\rbrace\\
&~~~~~~=\sum_{\gamma\in \Gamma_T}\max_{B^{0}_{lk}\in (0,1)}\left\lbrace n_\gamma\hat{B}^{(0)}_{\gamma}\log(B^{0}_{\gamma})+(1-\hat{B}^{(0)}_{\gamma})n_\gamma\log(1-B^{0}_{\gamma})\right\rbrace.
\end{align*}
We then conclude that the MLE under the null hypothesis of a hierarchical HSBM takes the form $\hat{B}^{0}=[\hat{B}^{0}_{lk}]$, where, if $(l,k)\in \gamma$ then
\begin{equation*}
\hat B^{0}_{lk}=\hat B^{0}_{\gamma}=\frac{1}{n_\gamma}\sum_{i<j}A_{ij}\delta_\gamma(i,j)
\end{equation*}

Under the alternative, the model is a traditional SBM, and the form of MLE is shown in (for example) \cite{bickel2009nonparametric}.
\end{proof}

\section{Proof of Theorem \ref{thm4}}
\label{sec:thm4pf}
Before proving Theorem \ref{thm4}, we first restate the Theorem for ease of reference.
\vspace{2mm}

\noindent\textbf{Theorem \ref{thm4}:} \emph{Let $\mathbb{P}_{0,S}$ indicate the probability conditioned on the null hypothesis in \ref{eqn:HSBMvar} being the true distribution, 
where these $\{s_{ij}\}$ further satisfy 
\begin{align}
\label{eq:S}
\epsilon_{S}=\min_{\gamma\in\Gamma}\min_{(i,j)\in\gamma_B}\left|s_{ij}-\frac{1}{|\gamma_B|}\sum_{(\ell,h)\in\gamma_B}s_{\ell h}\right|>0.
\end{align}
Let
\begin{align*}
    -2\hat\lambda_{\gamma}=\sum_{(l,k)\in \gamma_B}n_{lk}\log\left(\frac{1-\hat B_{lk}^{(1)}}{1-\hat B_{\gamma}^{(0)}}\right)+n_{lk}\hat B_{lk}^{(1)}\log
\left(\frac{\hat B_{lk}^{(1)}(1-\hat B_{\gamma}^{(0)})}{\hat B_{\gamma}^{(0)}(1-\hat B_{lk}^{(1)})}\right)
\end{align*}
We then have that $$
\mathbb{P}_{0,S}\left(
\frac{\epsilon_S^2n_\gamma}{9}\leq 
-2\hat\lambda_{\gamma}
\leq \frac{8 n_\gamma}{\delta}
\right)
\geq 1-2e^{-2(\min(\epsilon_S/3,\delta/2))^2n_\gamma}+2\sum_{(l,k)\in \gamma_B}e^{-2(\min(\epsilon_S/3,\delta/2))^2n_{lk}}
$$.}

\begin{proof} Let $\E_{0,S}[\cdot]$ (resp., $\mathbb{P}_{0,S}[\cdot]$) be the expectation (resp., probability) conditioned on the null from the hypotheses in Equation~\eqref{eqn:HSBMvar} and the collection of $\{s_{ij}\}_{(i,j)\in\gamma_B}$ being observed for each $\gamma\in\Gamma$ and that these $s_{ij}$'s satisfy Eq. \ref{eq:S}.
For a single collection of blocks $\gamma\in \Gamma$ and $(l,k)\in \gamma$, we have:
    \begin{align*}
        \E_{0,S}[\hat B^{(1)}_{lk}]=&B^{(0)}_{\gamma}+s_{lk}\\
        \E_{0,S}[\hat B^{(0)}_{\gamma}]=&B^{(0)}_{\gamma}+S_\gamma
    \end{align*}
    Where $S_\gamma=\frac{1}{|\gamma_B|}\sum_{(l,k)\in \gamma_B}s_{lk}$. 
    From Hoeffding's inequality applied to $\hat B^{(0)}_{\gamma}$ and $\hat B^{(1)}_{lk}$, we have that for any $t>0$, 
    \begin{align*}
         \mathbb P_{0,S}&\left(\left|\hat B^{(0)}_{\gamma}-B^{(0)}_{\gamma}-S_\gamma\right|\geq t\right)+\sum_{(l,k)\in \gamma_B}\mathbb P_{0,S}\left(\left|\hat B^{(1)}_{lk}-B^{(0)}_{\gamma}-s_{lk}\right|\geq t\right)\\
        & \leq  2e^{-2t^2n_\gamma}+2\sum_{(l,k)\in \gamma_B}e^{-2t^2n_{lk}}
    \end{align*}
    Let $t=\min(\epsilon_S/3,\delta/2)$, and note that 
    \begin{align}
        \label{ineq:Bhat1}
        \left|\hat B^{(1)}_{lk}-B^{(0)}_{\gamma}-s_{lk}\right|\leq  t,\quad\text{ and }\quad
        \left|\hat B^{(0)}_{\gamma}-B^{(0)}_{\gamma}-S_{\gamma}\right|\leq  t,
    \end{align}
    holds with probability at least 
    $$1-2e^{-2n_\gamma t^2}-
    2\sum_{(l,k)\in \gamma_B}e^{-2n_{lk}t^2}.$$
    Given the event in Eq. \ref{ineq:Bhat1},
    \begin{align*}
        \hat B^{(1)}_{lk}-B^{(0)}_{\gamma}-s_{lk}\geq -t;&\quad
        -\hat B^{(0)}_{\gamma}+B^{(0)}_{\gamma}+S_{\gamma}\geq -t;\\
            \hat B^{(1)}_{lk}-B^{(0)}_{\gamma}-s_{lk}\leq t;&\quad
        -\hat B^{(0)}_{\gamma}+B^{(0)}_{\gamma}+S_{\gamma}\leq t.
    \end{align*}
    Combining these inequalities, we see that 
    \begin{align*}
      &-2\epsilon_S/3-S_{\gamma}+s_{lk}\leq  \hat B^{(1)}_{lk}-\hat B^{(0)}_{\gamma}\leq 2\epsilon_S/3-S_{\gamma}+s_{lk}\\
      &\Rightarrow 
      \begin{cases}
      \hat B^{(1)}_{lk}-\hat B^{(0)}_{\gamma}\leq -\epsilon_S/3\,\,&\text{ if }S_{\gamma}>s_{lk}\\
      \hat B^{(1)}_{lk}-\hat B^{(0)}_{\gamma}\geq \epsilon_S/3\,\,&\text{ if }S_{\gamma}<s_{lk}\\
      \end{cases}
    \end{align*}
    and hence $|\hat B^{(1)}_{lk}-\hat B^{(0)}_{\gamma}|\geq \epsilon_S/3$.
    Applying \ref{eq:pinsker} with $P=\hat B^{(1)}_{lk},Q=\hat B^{(0)}_{\gamma}$,
    \begin{align*}
        -2\hat\lambda_{\gamma}=&\sum_{(l,k)\in \gamma_B}n_{lk} D_{\text{KL}}(P\,||\,Q)\geq \sum_{(l,k)\in \gamma_B}n_{lk}\left(\hat B^{(1)}_{lk}-\hat B^{(0)}_{\gamma}\right)^2\geq n_{\gamma}\epsilon_S^2/9
    \end{align*}
    Moreover, given the event in Eq. \ref{ineq:Bhat1}, we have that 
    $\hat B^{(0)}_{\gamma}\in(\delta/2,1-\delta/2)$
    so that Eq. \ref{eq:pinsker} also  provides
    \begin{align*}
        -2\hat\lambda_{\gamma}\leq &\sum_{(l,k)\in \gamma_B}\frac{4}{\delta} n_{lk}\left(\hat B^{(1)}_{lk}-\hat B^{(0)}_{\gamma}\right)^2\leq \frac{8}{\delta}n_\gamma
    \end{align*}
    as desired.
\end{proof}

\section{Proof of Theorem \ref{thm:penLR1}}
\label{sec:penLR1}

Consider the case where the graph under consideration comes from an RMHSBM distribution; note that below we will assume that the assumption in Eq.\@ \ref{eq:ass1} holds.
In this setting, an application of Theorem \ref{thm:chung} yields that for $n>n_0$ (letting $t=\sqrt{c_2 B_{\ell,k}^{(1)}\frac{\log n_{\ell,k}}{n_{{\ell,k}}}}$ for a constant $c_2>0$ to be set shortly)
\begin{align}
\label{eq:hsbm_sbmlower}
\mathbb{P}(\hat B_{\ell,k}^{(1)}- B_{\ell,k}^{(1)}\leq -t)&\leq 
e^{-\frac{c_2 \log n_{\ell,k}
}{2}}\\
\label{eq:hsbm_sbmupper}
\mathbb{P}(\hat B_{\ell,k}^{(1)}- B_{\ell,k}^{(1)}\geq t)&\leq 
\text{exp}\left\{
-\frac{c_2 B_{\ell,k}^{(1)}\log n_{\ell,k}}{
2B_{\ell,k}^{(1)}+
\frac{2}{3}\sqrt{c_2 B_{\ell,k}^{(1)}\frac{\log n_{\ell,k}}{n_{\ell,k}}}}
\right\}
\leq 
\text{exp}\left\{
-\frac{c_2 \log n_{\ell,k}}{
2+
\frac{2}{3}\sqrt{ c_2}}\right\}
\end{align}
Note that, if
\begin{align}
\label{eq:sbm_in_t}
    |\hat B_{\ell,k}^{(1)}- B_{\ell,k}^{(1)}|\leq \sqrt{c_2 B_{\ell,k}^{(1)}\frac{\log n_{\ell,k}}{n_{{\ell,k}}}}\text{ for all }\{\ell,k\}\in\gamma,
\end{align} 
then it holds that
$$
|\hat B_{\gamma}^{(0)}- B_{\ell,k}^{(1)}|\leq \max_{(\ell,k)\in\gamma_B} \sqrt{c_2 B_{\ell,k}^{(0)}\frac{\log n_{\ell,k}}{n_{{\ell,k}}}}
$$
as well.
Now, (where $n_{**}=\min_{\ell,k}n_{\ell,k}$)
\begin{align}   \mathbb{P}&\left(\bigcap_{\gamma\in\Gamma}\bigcap_{(\ell,k)\in\gamma_B}\left\{|\hat B_{\ell,k}^{(1)}- B_{\ell,k}^{(1)}|\leq \sqrt{c_2 B_{\ell,k}^{(1)}\frac{\log n_{\ell,k}}{n_{{\ell,k}}}} \right\} \right)\notag\\
&=\prod_{\gamma\in\Gamma}\prod_{(\ell,k)\in\gamma_B} \mathbb{P}\left(|\hat B_{\ell,k}^{(1)}- B_{\ell,k}^{(1)}|\leq \sqrt{c_2 B_{\ell,k}^{(1)}\frac{\log n_{\ell,k}}{n_{{\ell,k}}}}  \right)\notag\\
&\label{eq:sbm_alpha1}
\geq \prod_{\gamma\in\Gamma}\prod_{(\ell,k)\in\gamma_B} \left(1-2\text{exp}\left\{
-\frac{c_2 B_{\ell,k}^{(1)}\log n_{\ell,k}}{
2B_{\ell,k}^{(1)}+
\frac{2}{3}\sqrt{c_2 B_{\ell,k}^{(1)}\frac{\log n_{\ell,k}}{n_{\ell,k}}}}
\right\}\right)\\
&\label{eq:sbm_alpha2}
=\Omega\left( \text{exp}\left(-((K^*)^2+K^*)\cdot\text{exp}\left\{
-\frac{c_2 \log n_{**}}{
2+
\frac{2}{3}\sqrt{c_2}}
\right\} \right)\right)
\end{align}
We then have that with probability at least the value in Eq.\@ \ref{eq:sbm_alpha1}, it uniformly holds that (where $n_{*,\gamma}=\min_{\{\ell,k\}\in\gamma}n_{\ell,k}$)
\begin{align}
\label{eq:concB0}
\hat B_{\ell,k}^{(1)}&\in B_{\ell,k}^{(1)}\pm \sqrt{c_2 B_{\ell,k}^{(1)}\frac{\log n_{\ell,k}}{n_{{\ell,k}}}}\\
\label{eq:concB1}
\hat B_{\gamma}^{(0)}&\in B_{\ell,k}^{(1)}\pm \sqrt{c_2 B_{\ell,k}^{(1)}\frac{\log n_{*,\gamma}}{n_{{*,\gamma}}}}
\end{align}
Note that Eq. \ref{eq:concB1} implies that
\begin{align*}
\hat B_{\gamma}^{(0)}&\geq B_{\ell,k}^{(1)}- \sqrt{c_2 B_{\ell,k}^{(1)}\frac{\log n_{*,\gamma}}{n_{*,\gamma}}}
\end{align*}
and 
\begin{align*}
1-\hat B_{\gamma}^{(0)}&\geq 1-B_{\ell,k}^{(1)}- \sqrt{c_2 B_{\ell,k}^{(1)}\frac{\log n_{*,\gamma}}{n_{*,\gamma}}}
\end{align*}
Assuming Eq. \ref{eq:concB0}--\ref{eq:concB1} hold moving forward, applying the Pinsker upper bound in Eq. \ref{eq:pinsker} to Eq. \ref{eq:llr}, for $n>n_0$ we then arrive at the following which holds with probability at least $1-4K^2e^{-c_4\log n_*}$.
\begin{align*}
-2\hat\lambda_T&\leq \sum_{\gamma\in\Gamma}\sum_{(\ell,k)\in \gamma_B}n_{\ell,k}
\frac{2}{\min(\hat B_{\gamma}^{(0)},1-\hat B_{\gamma}^{(0)})}(\hat B_{\gamma}^{(0)}-\hat B_{\ell,k}^{(1)})^2\\
&\leq \sum_{\gamma\in\Gamma}\sum_{(\ell,k)\in \gamma_B}n_{\ell,k}
\frac{8c_2B_{\ell,k}^{(1)}\frac{\log n_{\ell,k}}{n_{{\ell,k}}}}{
B_{\ell,k}^{(1)}- \sqrt{c_2 B_{\ell,k}^{(1)}\frac{\log n_{*,\gamma}}{n_{*,\gamma}}}
}
\quad\text{(applying Eqs.\@ \ref{eq:concB0} and \ref{eq:concB1})}\\
&\leq \sum_{\gamma\in\Gamma}\sum_{(\ell,k)\in \gamma_B}
\frac{8c_2}{(1- \sqrt{c_2})}\log n_{\ell,k}
\end{align*}
We then have the BIC difference is equal to
\begin{align}
\notag
\hat\Delta_{T,BIC}&=-2\hat\lambda_T-\left(\sum_{\gamma\in\Gamma}(|\gamma_B|-1)\right)\log \binom{n}{2}\\
&\leq
\sum_{\gamma\in\Gamma} \left(|\gamma_B| 
\frac{8c_2}{(1- \sqrt{c_2})}-(|\gamma_B|-1)\right)\log \binom{n}{2}\\
\notag
&=\left[\left(\frac{8c_2}{(1- \sqrt{c_2})}-1\right)\frac{(K^*)^2+K^*}{2}+|\Gamma|\right]\log \binom{n}{2}
\label{eqn:LLRBIC}
\end{align}
Choosing $c_2\leq\frac{1}{16}\left(1-\frac{2|\Gamma|}{(K^*)^2+K^*}
\right)$
yields  
$\hat\Delta_{T,BIC}<0$ as desired.

\section{Proof of Theorem \ref{thm:penLR2}}
\label{sec:penLR2}

Again, applying Theorem \ref{thm:chung}, we have that with probability at least $1-4(K^*)^2e^{-\frac{c_2 \log n_{**}}{
2+
\frac{2}{3}\sqrt{c_3}}}$
it uniformly holds that 
\begin{equation} \label{eq:highprobevent} \begin{aligned}
\Bhat_{\ell,k}^{(1)}&\in B_{\ell,k}^{(1)}\pm \sqrt{c_3 B_{\ell,k}^{(1)}\frac{\log n_{\ell,k}}{n_{\ell,k}}}~\text{ and }\\
\Bhat_{\gamma}^{(1)}&\in \mathbb{E}(\Bhat_{\gamma}^{(1)})\pm \sqrt{c_3 \mathbb{E}(\Bhat_{\gamma}^{(1)})\frac{\log n_{\gamma}}{n_{\gamma}}},
\end{aligned} \end{equation}
where $n_\gamma=\sum_{(\ell,k)\in\gamma_B}n_{\ell,k}$.

If the RMHSBM is $(M,\eta)$-incompatible with the true SBM, then 
for $n>n_0$, applying the lower Pinsker bound in Equation~\eqref{eq:pinsker} to Equation~\eqref{eq:llr} yields (we remind the reader that $C>0$ is a constant which may change from line to line)
\begin{align*}
-2\lambdahat_T&\geq \sum_{\gamma\in\Gamma}\sum_{(\ell,k)\in \gamma_B}n_{\ell,k}
4(\hat B_{\gamma}^1-\hat B_{\ell,k}^0)^2=\Omega\left(n_{**} M\eta^2  -C M\log n_{**}\right) 
\end{align*}
Hence, we have the BIC difference satisfies
\begin{align*}
\hat\Delta_{T,\BIC}&=-2\hat\lambda_T-\sum_{\gamma\in\Gamma}(|\gamma_B|-1)\log \binom{n}{2}\\
&=\Omega\left(n_{**} M\eta^2  -C\left(\binom{K^*}{2}+K^*+M\right)\log n\right)
\end{align*}
Hence, if $\eta^2 M =\omega( (K^*)^2\frac{\log n}{n_{**}})$, under the high-probability event in Equation~\eqref{eq:highprobevent}, we have 
$\hat\Delta_{T,\BIC}>0$, completing the proof.

\end{document}